\documentclass[letterpaper,twocolumn,10pt]{article}
\usepackage{usenix}
\usepackage{tikz}
\usepackage{amsmath}
\usepackage{booktabs}
\usepackage{tabularx}
\usepackage{amssymb}
\usepackage{pifont}
\usepackage{graphicx}
\usepackage{filecontents}
\usepackage{xspace}
\usepackage{hyperref}
\usepackage{seqsplit}
\usepackage{float}
\usepackage{algorithm}
\usepackage{algpseudocode}
\usepackage{multirow}
\usepackage{todonotes}
\usepackage{balance}
\usepackage{xurl}
\usepackage{listings}
\usepackage{xcolor}
\usepackage{enumitem}
\usepackage{siunitx}
\usepackage{booktabs}
\usepackage{graphicx}
\usepackage{cuted}
\usepackage{tabularx}
\usepackage{booktabs}
\usepackage{caption}

\hypersetup{
    colorlinks=true,  
    linkcolor=blue,  
    urlcolor=blue, 
    citecolor=blue 
}
\newcommand{\sysname}{\textsc{EVAge}\xspace}
\newcommand{\eat}[1]{}

\newcommand{\cy}{\todo[inline,author=Chenyuan,color=brown]}

\newcommand{\stitle}[1]{\vspace{1.2ex}\noindent{\bf #1}}
\newcommand{\etitle}[1]{\vspace{1ex}\noindent{\em\underline{#1}}}
\newcommand{\tbf}{{\textcolor{red}{XX}\xspace}}

\begin{document}

\title{\Large \bf{\sysname: Autonomous MEV Generation and Adaptation via Multi-Agent Harness}}


\author{
{\rm Yan Wen\thanks{These authors contributed equally to this work.}}\\
{\normalsize City University of Hong Kong}
\and
{\rm Zichun Cai\footnotemark[1]}\\
{\normalsize City University of Hong Kong}
\and
{\rm Iliya Mirzaei}\\
{\normalsize Stony Brook University}
\and
{\rm Xiaohua Cai}\\
{\normalsize Tsinghua University}
\and
{\rm Mohammad Javad Amiri}\\
{\normalsize Stony Brook University}
\and
{\rm Haoxian Chen}\\
{\normalsize Shanghai Tech University}
\and
{\rm Chenyuan Wu}\\
{\normalsize City University of Hong Kong}
}

\maketitle

\begin{abstract}
Maximal Extractable Value (MEV) has evolved into a major economic force in blockchain ecosystems, yet its capture is dominated by experienced teams, and both strategy design and implementation rely on manual expert work that scales poorly across heterogeneous protocols and chains.
We present \sysname, the first fully autonomous multi-agent framework for end-to-end MEV strategy generation and adaptation. Equipped with three specialized operation modes, it automatically discovers novel MEV variants, adapts execution logic across disparate protocols, and ports strategies between chains, including Layer-1 and Layer-2 networks. To avoid inference latency on the critical MEV execution path, \sysname generates and refines MEV bot code offline rather than making real-time decisions directly. Under the coordination of an orchestrator agent, three specialized subagents collectively implement and repair the full MEV bot workflow via closed-loop diagnostics, eliminating human intervention while producing validated and deterministic Proof-of-Concept implementations. We evaluate \sysname on over 1.5M blocks from each of Ethereum, Base, and BNB Smart Chain (BSC). On Ethereum, \sysname uncovers five novel MEV strategy variants, yielding a profit increase of 1.02$\times$ to 15.97$\times$. It also successfully adapts 11 MEV strategies from CPMM to both CLMM and Balancer V2 and ports strategies from Ethereum to Base and BSC, all with less than $60$ dollars in LLM token costs.
\end{abstract}
\section{Introduction}

Maximal Extractable Value (MEV) has emerged as a prominent method of profit generation on  blockchains~\cite{torres2021frontrunner,daian2020flashboys,zhou2021just}.
On Ethereum, MEV enables searchers, builders, and validators to capture billions of dollars annually~\cite{qin2022quantifying,mclaughlin2023large,weintraub2022flash,wu2023defiranger,li2023demystifying}.
When MEV becomes large relative to block rewards, it can even destabilize consensus, e.g., by incentivizing chain reorganizations. Given these risks, substantial efforts have been made toward MEV mitigation in academia~\cite{kavousi_et_al:LIPIcs.OPODIS.2025.36,cryptoeprint:2026/1195,jiang2025armm,308144}, whereas the Ethereum community instead favors facilitating MEV extraction in an efficient, decentralized, and transparent manner rather than eliminating it~\cite{yang2023sokmevcountermeasurestheory}. 

In practice, block producers typically do not actively identify extractable value opportunities themselves. Instead, this task is outsourced to {\em searchers}, i.e., specialized traders running automated algorithms that continuously monitor the public mempool~\cite{owocki_mev_2026}.
Upon detecting a profitable opportunity among pending transactions, searchers incentivize block producers, via elevated transaction fees or direct bribes, to prioritize and order their transactions favorably within the block \cite{binance_academy_mev_2026}.

Despite the growth of this ecosystem, MEV extraction faces three major practical issues. First, MEV revenue accrues disproportionately to institutional players due to the substantial computational resources and technical expertise required to capture MEV~\cite{carlsten2016instability}. This concentration of MEV extraction fundamentally conflicts with the decentralization and transparency that blockchains are meant to provide~\cite{yang2023sokmevcountermeasurestheory,buterin2021proposer,hasu2022whyrunmevboost}. Experienced quantitative teams dominate this space, writing sophisticated scripts to monitor unconfirmed transactions in the public mempool using high-frequency on-chain trading strategies~\cite{luo2025rich}, identify arbitrage opportunities across liquidity pools~\cite{flashbots2020frontrunning,wang2021cyclic}, and backrun large swaps to capture transient price dislocations~\cite{qin2022quantifying}. As a result, ordinary users are structurally disadvantaged under this centralized MEV paradigm, bearing extraction costs without the technical expertise to compete for the opportunities themselves.

Second, existing MEV development practices do not scale well to the customized mathematical models underlying modern DeFi protocols. To accurately calculate execution slippage and optimize trade routing, searchers must meticulously work on each new protocol's logic. For instance, simulating arbitrary trades on Uniswap V3, a Concentrated Liquidity Market Maker (CLMM), requires precise off-chain reconstruction of tick map distributions, fee tiers, and intricate tick-crossing mechanisms~\cite{adams2021uniswap}. Manual scripts often struggle to handle this state space, resulting in transaction reverts and substantial gas losses.
Likewise, in Balancer V2, transactions must interface with a unified Vault rather than separate contracts for each trading pair, decoupling asset management from the protocol's customized weighted math~\cite{martinelli2021balancerv2}. Searchers must carefully validate off-chain simulations before constructing the \texttt{BatchSwapStep} structs required for \texttt{batchSwap} execution~\cite{balancer_batchswaps}. Mastering such protocol-specific mathematics and complex interfaces is non-trivial, yet these challenges are rarely addressed in the literature, let alone for the long tail of less popular protocols across modern blockchains~\cite{bingx_humidifi, euler_swap, bunni_v2}. By comparison, Ladóczki et al.~\cite{ladoczk2026does} exclusively examine Layer-1 Constant Product Market Makers (CPMMs), explicitly omitting CLMMs, while Gogol et al.~\cite{gogol2025priority} likewise illustrate only CPMM dynamics. The challenge of tailoring MEV extraction strategies to such diverse and customized protocols thus remains understudied, even as the complexity of the DeFi ecosystem continues to grow rapidly.

Finally, manual scripting introduces severe adaptation latency when migrating MEV strategies across different chains. The launch of new blockchains and Layer-2 rollups creates a brief ``vacuum period'' characterized by low competition and exceptionally high MEV profitability. For instance, MEV searchers adapted to zkSync within just one month of its developer release~\cite{ferreira2024rolling}. A similar case is the race to integrate Arbitrum's new TimeBoost sequencing auction~\cite{messias2025express}. Delayed deployment is costly for searchers, risking the loss of these critical early-mover advantages. How to rapidly and dynamically adapt MEV strategies to new environments was a practical problem long overlooked by the academic community.

To address the above MEV deployment and implementation challenges, we propose a fully autonomous MEV generation and adaptation system, \sysname, enabled by recent advances in AI reasoning and coding. Designing an effective AI agent framework for MEV generation, however, presents several unique challenges.
First, the high inference latency of agents conflicts with the latency-sensitive nature of MEV execution. Second, MEV generation requires exploring an unstructured and dynamically evolving search space. Unlike smart contract vulnerability exploitation, where agents operate within well-defined, constrained sandboxes~\cite{gervais2026aiagentsmartcontract}, profitable MEV strategy variants can rarely be specified in advance.
The continuous emergence of distinct DeFi protocols and chains further requires the system to avoid any environment-specific assumptions. Third, modern LLMs are prone to hallucinations, while generated runnable Proof of Concepts (PoCs) must be correct and deterministic. Directly executing AI-generated code may expose users to substantial financial losses.

\sysname addresses these challenges through a multi-agent orchestration harness~\cite{li2026agent}. To avoid inference latency on the critical MEV execution path, agents generate and refine MEV bot code offline rather than make real-time decisions directly by themselves. 
To assist agents in exploring the unstructured MEV space, we decompose the MEV bot workflow into distinct stages and assign clear responsibilities to specialized subagents, coordinated by a tailored closed-loop harness. The harness routes operation control among agents, performs multi-level evaluation and precise failure attribution, and returns targeted diagnostic feedback to the responsible agent.
To ensure safe verification and mitigate hallucination risks, we decouple the agentic reasoning layer from the MEV bot's deterministic workflow. Rather than allowing an LLM to generate an entire bot from scratch, each agent's intervention is confined to a narrow MEV workflow stage determined by its specialized role (Fig.~\ref{fig:agent_loop_full}). After each modification, the generated MEV bot must be simulated by the validator, a deterministic caller of anvil that operates on a forked block state. Agents do not directly produce the final runnable PoC; instead, the PoC is strictly derived from the validator's concrete execution trace. Consequently, the validator can intercept any hallucinated logic before it propagates to a live deployment.

Through this design, \sysname shifts AI agents from reactive analysis tools to proactive MEV bot generators, reducing the engineering burden and implementation barriers to MEV extraction. \sysname supports three core functionality modes. In variant explorer mode, it discovers new MEV variants and improves existing strategies. In protocol adaptation mode, it adapts MEV workflows to unseen AMMs with different mathematical models and interfaces. In cross-chain adaptation mode, it migrates strategies across heterogeneous execution environments, including Ethereum, Base, and BNB Smart Chain (BSC). Our key contributions are:

\begin{itemize}
\item \textbf{Autonomous MEV generation.} To our knowledge, \sysname is the first agentic system to autonomously generate MEV bots without human intervention, validated through extensive experiments on over $1.5$ million historical blocks. \sysname autonomously synthesizes insights from existing literature, generates end-to-end implementations of strategies and their variants, and adapts them to new protocols and alternative blockchains. It shifts MEV strategy maintenance and upgrades from a manual engineering bottleneck to a scalable, agentic paradigm.

\item \textbf{Hierarchical multi-agent harness for MEV.} We propose a novel multi-agent orchestration framework that strictly decouples the MEV bot into three specialized roles (Collector, Builder, and Validator; Fig.~\ref{fig:agent_loop_full}), significantly reducing each agent's context burden and mitigating hallucinations. By combining closed-loop diagnostics with goal-oriented natural-language handoffs, our architecture generates MEV bot implementations automatically while guaranteeing economic safety.

\item \textbf{Cross-protocol and cross-chain adaptation.} \sysname enables seamless cross-protocol and cross-chain migration. Its modular MEV bot architecture isolates most protocol- and chain-specific changes to the collector and validator, allowing the core planner logic to be reused despite divergent on-chain data models and execution environments.
Accordingly, the adapted implementations achieve 18.13\% average structural code similarity, concentrated primarily in the planner.

\item \textbf{Discovery of novel MEV variants.} We identify five novel MEV strategy variants for community reference, demonstrating \sysname's proficiency in massive search-space exploration and variant parameter optimization. We highlight that \sysname independently provides a mathematical proof of SBA-HFT's superiority (Appendix~\ref{app:SBA-HFT_theoretical_proofs}). In our experiments, the discovered variants increase revenue by $1.02\times$ to $15.97\times$ relative to their corresponding baseline strategies.

\item \textbf{Prototype implementation and comprehensive evaluation.} We implement a prototype of \sysname and evaluate it on over $1.5$ million historical blocks collected from February to August 2025. Our evaluation covers 11 MEV strategies on Ethereum, six on BNB Smart Chain, and two on Base. The results demonstrate profitable MEV execution, practical token costs, low planning latency, and the effectiveness of the multi-agent harness.
\end{itemize}
\section{Background}
\label{Sec:Background}

\stitle{AMM.}
Decentralized Finance (DeFi) operates on blockchain networks via smart contracts, which are deterministic state-transition functions running on virtual machines. Network computing resources are allocated through a gas fee mechanism. Following the EIP-1559~\cite{eip1559} upgrade, this pricing model includes a base fee and a priority fee, the latter of which frequently triggers Priority Gas Auctions as users and bots compete for favorable transaction sequencing \cite{daian2020flashboys}. Instead of relying on traditional order books, DeFi heavily utilizes Automated Market Makers (AMMs) to facilitate continuous, permissionless trading through liquidity pools. The most prevalent AMM model is the Constant Product Market Maker (\textbf{CPMM}). Assuming a token pair $(X, Y)$ traded on a CPMM with reserves $(x, y)$ and fee $f \in [0, 1)$, where $x$ and $y$ represent the reserve balances of the two tokens $(X, Y)$ in the liquidity pool. If a trader swaps $q > 0$ units of $X$ with the CPMM pool, the trader receives~\cite{gogol2025priority}:

$$ \Delta y(q) = \frac{y(1 - f)q}{x + (1 - f)q} $$
\label{Formula:CPMM}

and the reserves in the CPMM pool update to $(x + (1 - f)q, y - \Delta y(q))$. Under the CPMM pricing model, large trades trigger severe price impact and slippage. This pricing mechanism allows malicious actors to predict markets for exploitation.
Another important AMM is the Concentrated Liquidity Market Maker (\textbf{CLMM}), which replaces a single global reserve curve with concentrated liquidity bounded by discrete price ranges known as \emph{ticks}. At any given tick $t$, the raw token price is proportional to $1.0001^t$.
For computational efficiency on-chain, the system tracks the square root of the price, $\sqrt{P}$, where $P$ is the relative price of the two tokens. 

During an intra-tick trade of size $q$ with fee factor $f$, the virtual liquidity $L = \sqrt{x \cdot y}$ remains constant, and the price square root transitions from $\sqrt{P}$ to $\sqrt{P'}$:
$$\sqrt{P'} = \frac{L \sqrt{P}}{L+(1-f) q \sqrt{P}}$$
The trader receives a corresponding token output of:
$$F(q) = L(\sqrt{P}-\sqrt{P'})$$
\label{eq:clmm-output}
These formulas apply only within the current tick interval. If a large trade pushes the price across the boundary into a new tick $t$, the pool must adjust its active liquidity based on the funds in that next bracket. Let $\mathrm{LiquidityNet}(t)$ denote the net liquidity added to or removed from the pool at the boundary of tick $t$; the new active liquidity is:$$L_{\mathrm{next}} = L_{\mathrm{current}} - \mathrm{LiquidityNet}(t)$$

\paragraph{Maximal Extractable Value (MEV).}
Maximal Extractable Value (MEV)~\cite{piet2022extracting,torres2021frontrunner,luo2026light,materwala2025maximal,zhang2023your} refers to the excess profit that block producers or third-party searchers can extract by manipulating the inclusion, exclusion, or ordering of transactions within a block. We focus on two primary MEV strategies: sandwich attack and arbitrage.
A {\em sandwich attack} occurs when an attacker spots a victim's pending large trade in the public mempool. The attacker frontruns the trade to artificially inflate the asset's price, causing the victim's trade to execute at (or near) its maximum acceptable slippage, then immediately backruns it to sell the asset at a risk-free profit~\cite{qin2022quantifying}.
{\em Arbitrage} is a MEV strategy in which searchers exploit temporary price discrepancies across different liquidity pools or decentralized exchanges. By executing atomic, cross-market trades, arbitrageurs profit while simultaneously helping to enforce price consistency across fragmented markets.

\stitle{MEV workflow.}\label{Sec:background_MEV_workflow}
While recent studies, e.g., APOLLO~\cite{luo2026light}, examine the broader lifecycle of MEV bots, such as a bot's market entry or exit, we instead focus on the direct mechanics of how a bot constructs and executes a single profitable trade. Consider a sandwich bot monitoring the Ethereum mempool as an example. When it spots a pending large user swap (e.g., USDC $\rightarrow$ WETH) that will push up the WETH price in a Uniswap pool, the bot quickly builds an execution plan: it borrows capital via a flashloan to front-run the swap by buying WETH first (e.g., 10 WETH), then plans to sell it back at the elevated price right after the user's trade executes. To finalize the plan, the bot sets precise parameters, calldata, slippage tolerance, and gas price, while weighing the risk that the user's transaction fails or a rival bot wins the same opportunity. Finally, it submits a bundle directly to a block builder, offering a share of the profit as a priority fee. This guarantees the builder executes the transactions atomically and in order: the bot's buy, then the user's swap, then the bot's sell, immediately followed by repayment of the flashloan.

\eat{
\stitle{Layer 1 vs. Layer 2.}
Unlike monolithic Layer 1 (L1) networks, which handle execution, consensus, and data availability together, Layer 2 (L2) rollups decouple these functions to achieve greater scalability. L2 networks typically rely on a centralized sequencer with a private mempool to order transactions and construct blocks \cite{ferreira2024rolling}. This design hides pending transactions from the public network, removing the visibility that traditional sandwich attacks require and pushing MEV searchers toward "optimistic MEV" strategies~\cite{solmaz2025optimistic, wang2026blockspace, gogol2025priority} based on high-frequency, speculative transactions. 
We focus our experiments on Base, which offers a particularly specialized MEV environment due to its unique block-building design. Base implements "Flashblocks," splitting the standard 2-second block time into 200-millisecond partial blocks, each governed by a localized priority fee auction \cite{wang_flashblocks_2025}. Once a Flashblock is built, its transaction ordering is permanently fixed, tightening the timing and execution constraints on MEV strategies. Despite these tight latency constraints, Base's near-zero transaction fees keep continuous, high-frequency arbitrage highly profitable.
}

\section{\sysname Functionality Overview}
Traditionally, the discovery and implementation of MEV strategies have been time-consuming and severely constrained by manual effort~\cite{qin2022quantifying,qin2023imitation,zhang2023your}. Developers struggle to manually maintain codebases across rapidly evolving DeFi protocols, increasingly complex economic mechanisms, and heterogeneous chains. \sysname addresses these limitations by transitioning from previous reactive approaches to a \textit{proactive} agentic pipeline, supporting three core functionalities: novel strategy exploration, protocol adaptation, and cross-chain adaptation. As illustrated in Fig.~\ref{fig:architecture_overview}, \sysname decomposes the MEV Bot workflow into three stages: the {\em Data Collector} that processes raw on-chain data and filters it down to actionable, exploitable opportunities; the {\em MEV Planner} that executes the core MEV logic, drafting concrete attack plans from these opportunities; and the {\em Validator} that assembles the corresponding transactions and submits the resulting MEV bundles to an upstream sequencer or builder.

Users can activate one functionality mode at a time by specifying a high-level goal to the orchestrator when launching \sysname (Fig.~\ref{fig:agent_loop_full}). Under different modes, the orchestrator will reconfigure the subagents and their feedback loops to adapt the relevant workflow stages. As indicated by the blue arrows in Fig.~\ref{fig:architecture_overview}, the MEV variant explorer mode primarily modifies the MEV planner to develop novel MEV logic, while the protocol adaptation mode adjusts the data collector to accommodate new protocol semantics and interfaces. We discuss the \sysname design in more detail in Sec.~\ref{sec:system}.

\begin{figure}[t]
    \centering
    \includegraphics[width=1\linewidth, trim=7cm 6cm 6cm 5cm, clip=true]{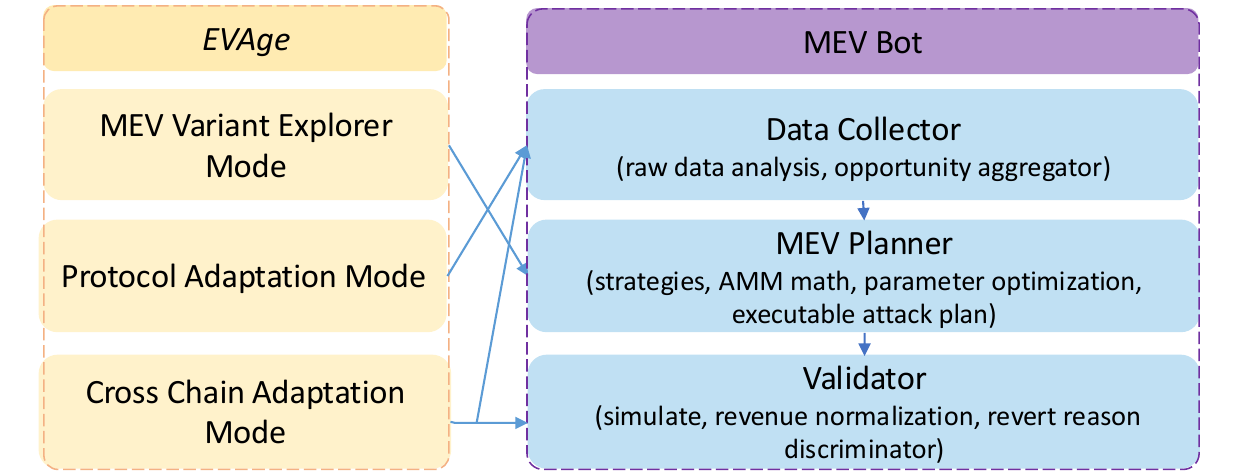}
    \caption{\sysname functionality overview. The three blue blocks on the right represent the stages of the MEV Bot workflow: data collector, planner, and validator. \sysname supports three complementary modes without human intervention: (1) \textbf{MEV Variant Explorer Mode} primarily targets the planner to synthesize novel strategy logic; (2) \textbf{Protocol Adaptation Mode} mainly focuses on adapting the collector to translate AMM semantics; and (3) \textbf{Cross-Chain Adaptation Mode} primarily adjusts the data collector to heterogeneous chains and makes sure the validator fits the target chain. The blue arrows indicate the MEV bot stages modified by each mode.}
\label{fig:architecture_overview}
\end{figure}

\stitle{MEV Variant Explorer.}
Relying on human intuition to discover new MEV strategies is slow and leaves promising regions of the vast strategy design space unexplored. To enable systematic exploration, \sysname introduces a variant explorer mode that synthesizes insights from existing literature to autonomously generate and evaluate novel MEV strategies. Rather than simply applying known methods to new environments, this mode focuses strictly on uncovering genuinely new attack opportunities. Our results (Sec.~\ref{Sec_Experiments: Variants}) show that \sysname discovers new effective strategies, most commonly through parameter optimization. 

\stitle{Protocol Migration.}
In the early stages of DeFi, MEV development was relatively straightforward, largely due to the structural homogenization of the market. As highlighted by~\cite{xu2023sok}, the majority of top AMMs, including SushiSwap, PancakeSwap, and QuickSwap, were fundamentally clones of the CPMM, retaining its core mathematical framework with only minor structural adjustments~\cite{angeris2021analysisuniswapmarkets}. Engineers face little pressure to adapt to new models, since liquidity was always distributed uniformly under the CPMM.
However, as the DeFi ecosystem has evolved, highly customized protocols with sophisticated mechanisms have proliferated in recent years~\cite{adams2021uniswap, martinelli2021balancerv2, adams2023uniswap, bingx_humidifi, euler_swap, bunni_v2}, making rapid adaptation to new protocol models a growing challenge for MEV searchers. Taking CLMMs as an example, liquidity is concentrated within specific price ranges rather than distributed uniformly, so each swap may traverse multiple ticks with different liquidity levels. Accurately calculating the output of a multi-hop arbitrage now requires navigating a highly fragmented state space.

\sysname mitigates this adaptation and maintenance burden via a fully automated protocol migration mode. It separates the core profit-making logic (e.g., temporarily providing liquidity to earn transaction fees) from its underlying implementation. In \sysname, specialized subagents alter the data collector (Fig.~\ref{fig:architecture_overview}) to represent pool states for different protocols, update the planner to handle heterogeneous victim constraints such as slippage tolerances and execution deadlines, and modify the validator to encode calldata, including function selectors and routing arrays. Consequently, this design preserves MEV strategy originality, while enabling their implementations to be systematically adapted across protocols. We evaluate \sysname's adaptation capability across two protocols, covering 11 strategies (Sec.~\ref{sec:rq2-protocol}). In our experiments, \sysname autonomously introduces tick-coverage frontiers to address the unique constraints of tick maps when full state snapshots cannot be reconstructed. It also gracefully adapts shared weighted-math modules to manage vault-mediated asset routing in Balancer V2. Across these adaptations, we find the code for core planner logic is consistently reused, substantially reducing implementation effort and token usage.

\stitle{Chain Adaptation.}
While it is vital to compete for the ``vacuum period'' following a new blockchain's launch, scaling an MEV operation from Ethereum to alternative Layer-1s or Layer-2 rollups is challenging. At the systems level, blockchain networks differ significantly in how they structure blocks and process transactions; at the economic level, they employ distinct fee mechanisms and auction models. For instance, Base charges separate L2 execution fees and L1 data availability gas fees~\cite{gogol2025priority}, and relies on a centralized sequencer with a private mempool. Manually porting and implementing separate MEV strategies for each new environment incurs substantial engineering and operational overhead.

To support automated deployment across different blockchains, \sysname introduces a cross-chain adaptation mode that preserves the core trading logic while adapting the underlying implementation details to each target blockchain. Specifically, the system dynamically adjusts how it reads market data, interacts with on-chain exchanges, calculates network fees, and constructs the final transactions for the target chain. We evaluate the chain adaptation mode on Base and BSC (Sec.~\ref{sec:rq3-chain}). Our experiments show that repeatedly adapting different strategies to the same chain incurs only marginal additional cost, since agents tend to reuse existing codebase.

\eat{\stitle{Applications.}}

\eat{
By lowering the engineering cost of strategy development and adaptation,
\sysname can support smaller searchers and DAOs in studying or internalizing
MEV. Protocol developers can also use it for adversarial testing on forked
states to identify harmful economic interactions before deployment. Finally,
its validated plans and execution traces provide a reproducible benchmark for
studying AI agents in DeFi.
}
\eat{(1) \sysname lowers the cost for smaller searchers and DAOs to study or
internalize MEV; (2) developers can test harmful economic interactions on
forked states before deployment; and (3) its traces benchmark AI agents in
DeFi.}
\eat{(1) MEV extraction is concentrated among teams that can maintain
protocol-specific searcher infrastructure~\cite{yang2023sokmevcountermeasurestheory};
by reducing this engineering burden, \sysname can help smaller searchers and
DAOs study or internalize MEV. (2) The same workflow can serve as adversarial
testing: protocol developers can search forked states for profitable, harmful
transaction sequences before deployment. (3) By recording validated plans,
code revisions, and execution traces, \sysname provides a reproducible testbed
for assessing AI agents in adversarial DeFi environments.}
\section{\sysname System Design} \label{sec:system}

\begin{figure*}[!t]
    \centering
    \includegraphics[width=1\textwidth, trim=4.0cm 1.8cm 8.5cm 1.5cm, clip=true]{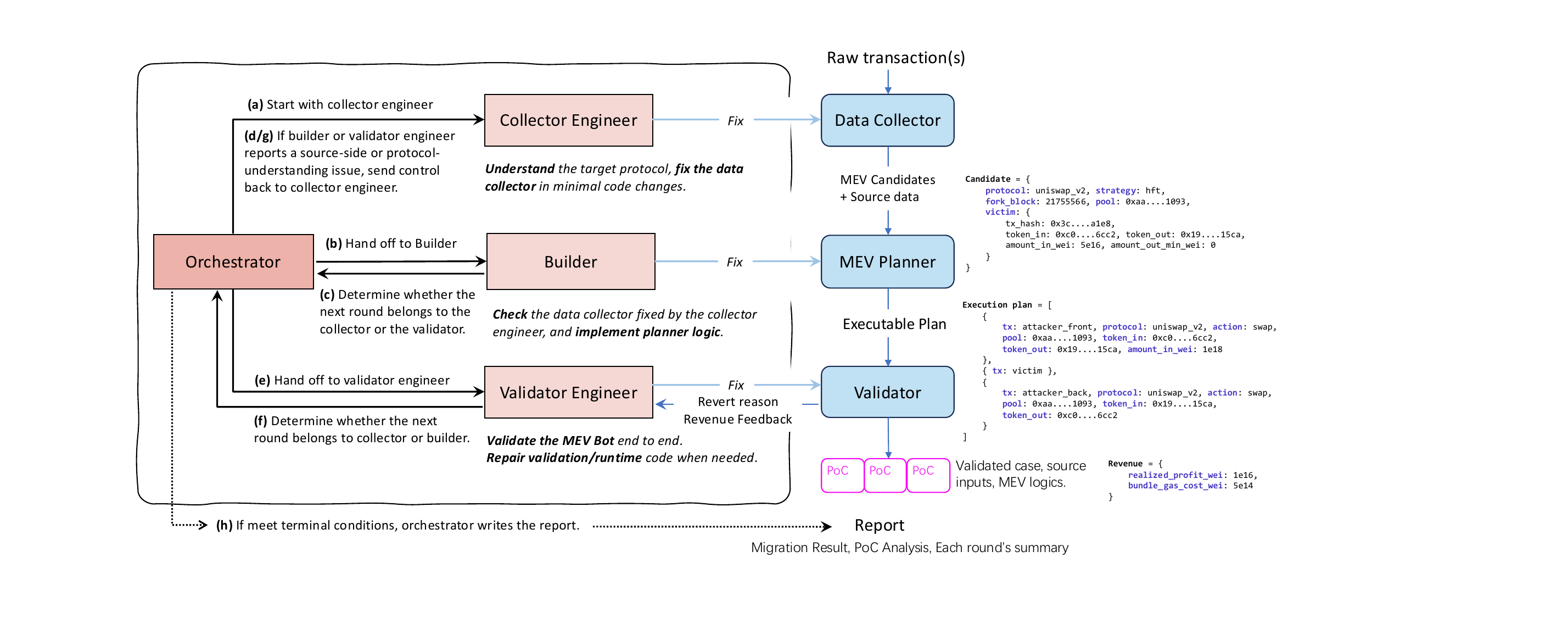}
    \caption{The system architecture of \sysname, detailing the internal workflow of the protocol adaptation mode. The parent orchestrator governs the sequential scheduling of subagents, incorporating diagnostic loopback routing for iterative control. Because \sysname executes only one dedicated mode at a time, this diagram serves as a representative example of \sysname's working modes.}
    \label{fig:agent_loop_full}
\end{figure*}

In this section, we present the system architecture of \sysname. As illustrated in Fig.~\ref{fig:agent_loop_full}, \sysname organizes its system components into two primary categories: (1) the \textbf{Agent Harness}, comprising an orchestrator and three specialized subagents responsible for coordination, reasoning, and code generation; and (2) the \textbf{MEV Bot}, comprising the data collector, MEV planner, and validator that form a three-stage execution workflow. The Agent Harness iteratively repairs the code implementing these stages, while the MEV Bot executes the resulting pipeline and provides concrete feedback. We avoid directly employing agents to control MEV strategies or make decisions in real time due to the high inference latency of LLMs. In contrast, this architectural decoupling removes LLM inference latency from the critical path of MEV execution while ensuring that, despite the nondeterministic nature of agent reasoning, the generated MEV bot executes deterministically.

\subsection{MEV Bot Decomposition}
\label{3-stage-seperation}
Assigning an entire complex MEV workflow to a single LLM can overload its context, leading to implementation errors or missed instructions. Inspired by prior work on multi-agents and task specialization~\cite{liu2024lost,qian2024chatdev}, \sysname strictly decomposes the MEV bot workflow into three distinct stages to limit the task complexity. Specifically, the data collector acquires and processes raw on-chain data (transactions and liquidity states) to aggregate candidate MEV opportunities. The planner performs the core mathematical computations and parameter optimization, applying AMM invariants to estimate profits. It then converts the computed results into a concrete plan that specifies the transactions, parameters, and execution conditions required for on-chain submission. Finally, the validator simulates the plan, normalizes revenue, and diagnoses transaction reverts or other failures. The core of the validator is a deterministic caller of Anvil, a local node with forking capabilities from an EVM-compatible chain developed by Foundry~\cite{foundry_anvil}.

As illustrated in Fig.~\ref{fig:agent_loop_full}, the separation establishes a clear division of stages, each assigned to a corresponding specialized agent. This task separation prevents overlapping edits and reduces the cognitive load on each subagent, mitigating the risk of LLM hallucinations. It also improves maintainability and code reuse by allowing common infrastructure, such as the core mathematical planner, to be shared seamlessly across different MEV strategies. More details about each stage of the MEV bot workflow are in Appendix~\ref{Appendix:3-stage-separation}.

\subsection{Hierarchical Multi-Agent Harness}
At the core of agent harness is an orchestrator agent that coordinates three specialized subagents without human intervention. The orchestrator operates in one mode at a time, with each mode defining a single objective, an initial subagent, and a set of permitted feedback paths. Specifically, cross-protocol and cross-chain migration initiates the process strictly with the collector engineer to prioritize source-side adaptation. In contrast, for MEV variant exploration, the focus shifts to the builder to generate novel MEV planner logic before iteratively refining the collector and validator. This mode-specific coordination gives each subagent a clear objective and confines its modifications to the corresponding MEV bot stage, thereby reducing instruction ambiguity.
Fig. \ref{fig:agent_loop_full} primarily elaborates the protocol adaptation mode, where an existing MEV strategy is migrated to a different DeFi protocol while preserving its core economic intent and control flow. The adaptation focuses on protocol-level differences, such as pool discovery, state representation, fee calculation, swap entry points, callback semantics, and liquidity or pricing mathematics.

\stitle{Workflow.}
Protocol adaptation follows a serial, iterative workflow governed by the orchestrator. As shown in Fig.~\ref{fig:agent_loop_full}, it begins when \textbf{(a)} the orchestrator assigns the initial task to the collector engineer, responsible for understanding the target protocol and generating empirical MEV candidates. The code snippet on the right-hand side of Fig.~\ref{fig:agent_loop_full} illustrates a typical candidate for a sandwich opportunity. Next, \textbf{(b)} the orchestrator transfers control to the builder. The builder then modifies the MEV planner, which uses these candidates to construct an executable attack plan. During this phase, \textbf{(c)} the builder assesses whether the plan can be made and determines which subagent should handle the next round. \textbf{(d \& g)} 
If the builder finds any mistakes in the collector logic or protocol understanding, the orchestrator will send the control back to the collector engineer to fix the root problem.
Once builder decides the plan compatibility is achieved, \textbf{(e)} execution is delegated to the validator engineer, who runs the entire MEV bot to evaluate profitability and capture execution traces. Following this, \textbf{(f)} if the validator determines the current bot cannot achieve positive profitability, control will be returned to the builder for refinement. These iterative loops continue until the bot workflow is validated with at least one confirmed profitable historical case, \textbf{(h)} at which point the orchestrator synthesizes a comprehensive report and submits the final code.

\stitle{System Outputs.}
Each run of \sysname produces three categories of outputs. First,
the agent-generated code. Through iterative closed-loop refinement, the
agents repair and produce executable MEV code that is either a novel MEV workflow or an adaptation to the target protocol or chain. Second, an experiment
report. The orchestrator produces a structured report that summarizes each loop iteration and states whether the variant exploration or adaptation was successful. Third, a Proof of Concept (PoC).
A PoC is a verifiable proof that the agent-generated code, when executed atomically on a forked block state, completes the entire attack flow successfully and yields a net positive revenue.
To prevent agents from fabricating results, 
agents do not produce PoC directly. Instead, the PoC is strictly derived from the execution trace of the validator.
While the agent may refine the validator by adjusting Anvil invocation methods or profit normalization, it cannot directly intervene in the Anvil execution or alter its results. Every generated PoC inherently contains an execution record, which guarantees the objectivity and stability of the feedback.

\stitle{The other two modes.} \label{Sec:other two modes design}
Compared with protocol adaptation, the MEV variant explorer mode initiates the workflow at the builder rather than the collector engineer. The builder first proposes and evaluates distinct strategy variants, implements the corresponding planner, and then incrementally modifies the collector and validator to support the promising variant. If a proposed mechanism is invalidated, control returns to the builder to explore a new branch. The exploration terminates only after \sysname identifies a variant that is distinct from existing strategies with PoCs.
In comparison, cross-chain adaptation mode is similar to cross-protocol, except for maintaining independent RPC connections for each chain.

\stitle{Closed-loop diagnostics.}
The structured handoffs among subagents establish a closed diagnostic loop, where subagents communicate execution states, EVM revert reasons, and profitable outcomes. Rather than relying on blind retries, \sysname routes specific errors back to the responsible agent (e.g., instructing the builder on why a transaction reverted) without cluttering the broader system context. To guarantee economic safety in variant explorer mode, \sysname offers proofs for verified profitable variants. This autonomous loop ultimately compiles all validated MEV strategies into a report and PoCs.

\subsection{Implementation}

Our system implementation consists of 5,367 lines of Python for six MEV bot workflows (HFT, SBA, JIT, BUR, LR, and MLL detailed in Tab.~\ref{tab:variants}) and 518 lines for the agent harness\footnote{Source code is at \url{https://anonymous.4open.science/r/Blockchain-MEV-CF60}.}. Our implementation enforces a strict architectural boundary between the agentic module and the MEV bot. Agents are only used for code generation, while historical replay and profitability evaluation are performed entirely by the offline MEV bot, which replays historical
blocks against forked Anvil nodes.
All agents run through the Codex CLI (v0.144.6), with an orchestrator agent session coordinating three specialist agents~\cite{openai_subagents}: collector engineer, builder, and validator engineer. 


Following Fig.~\ref{fig:agent_loop_full}, the MEV Bot framework follows a three-stage pipeline: (i)~a source layer (data collector) extracts transactions, receipts, logs, and historical protocol state and constructs MEV candidates; (ii)~strategy specific planners (planner) derive executable attack parameters; (iii)~a runtime engine (validator) validates each plan against an isolated Anvil fork and invokes a unified profitability module that values the attacker's token balance changes using historical CoinGecko prices and deducts transaction fees. Our implementation on Ethereum provides 33 strategy-protocol lanes, covering 11 MEV strategies across CPMM, CLMM, and Balancer V2. Additional chain-specific implementations support Base and BNB Smart Chain.
The Ethereum experiments query the local Reth endpoints, whereas the Base and BNB Smart Chain experiments use chain-specific archive RPC endpoints loaded from the experiment environment.
\sysname produces SQLite databases (recording each specific attack cases), checkpoints, summaries, and replayable source inputs.

\section{Experimental Evaluation}

\begin{table*}[t]
    \centering
    \caption{MEV reference strategies reproduced and variants generated by \sysname.
    The table summarizes their category and execution structures.
    The complexity labels \textit{Basic} and \textit{Complex} distinguish
    canonical single-victim strategies from strategies involving multiple
    victims or composite attacker actions, respectively.}
    \label{tab:variants}
    \scriptsize
    \begin{tabularx}{\textwidth}{@{}c l >{\hsize=0.9\hsize\raggedright\arraybackslash}X l >{\hsize=1.1\hsize\raggedright\arraybackslash}X}
        \toprule
        \textbf{Type} & \textbf{Complexity} & \textbf{Strategy} & \textbf{Category} & \textbf{Description} \\
        \midrule

        \multirow{6}{*}{Reference}
        &\multirow{3}{*}{Basic}
        & {\bf HFT}: High-Frequency Trading~\cite{zhou2021high} & Sandwich & 1$\times$Frontrun (optimal sizing) $\rightarrow$ 1$\times$Victim $\rightarrow$ 1$\times$Backrun\\
        && {\bf SBA}: Swap Backrun Arbitrage & Backrun & 1$\times$Victim $\rightarrow$ 1$\times$Backrun  \\
        && {\bf JIT}: Just-In-Time Liquidity~\cite{zhou2021just} & Liquidity & 1$\times$Add Liquidity $\rightarrow$ 1$\times$Victim $\rightarrow$ 1$\times$Remove Liquidity  \\
        \cmidrule{2-5}
        &\multirow{3}{*}{Complex}
        & {\bf BUR}: Burger Sandwich Attack~\cite{sandwichattack} & Sandwich & 1$\times$Frontrun $\rightarrow$ N$\times$Victims (same-direction swaps) $\rightarrow$ 1$\times$Backrun \\
        && {\bf LR}: Liquidity Refined~\cite{sandwichattack} & Sandwich & 1$\times$Frontrun $\rightarrow$ 1$\times$Victim $\rightarrow$ 1$\times$Backrun \textnormal{(Swap--Add Liquidity--Swap)} \\
        && {\bf MLL}: Multi-Layered Liquidity~\cite{sandwichattack} & Liquidity & 1$\times$Add Liquidity $\rightarrow$ N$\times$Victims $\rightarrow$ 1$\times$Remove Liquidity  \\
        \midrule

        \multirow{5}{*}{Variants}
        &\multirow{2}{*}{Basic}
        & {\bf SBA-HFT}: Swap Backrun Arbitrage with High-Frequency Trading Optimization & Backrun & 1$\times$Victim $\rightarrow$ 1$\times$Backrun (optimal sizing) \\
        && {\bf JIT-HFT}: Just-In-Time with High-Frequency Trading Optimization & Liquidity & 1$\times$Add Liquidity $\rightarrow$ 1$\times$Victim (optimal sizing) $\rightarrow$ 1$\times$Remove Liquidity \\
        \cmidrule{2-5}
        &\multirow{3}{*}{Complex}
        & {\bf BUR-HFT}: Burger Sandwich Attack with High-Frequency Trading Optimization & Sandwich & 1$\times$Frontrun (optimal sizing) $\rightarrow$ N$\times$Victims $\rightarrow$ 1$\times$Backrun\\
        && {\bf LR-HFT}: Liquidity Refined with High-Frequency Trading Optimization & Sandwich & 1$\times$Frontrun (optimal sizing) $\rightarrow$ 1$\times$Victim $\rightarrow$ 1$\times$Backrun \textnormal{(Swap--Add Liquidity--Swap)} \\
        && {\bf MLL-HFT}: Multi-Layered Liquidity with High-Frequency Trading Optimization & Liquidity & 1$\times$Add Liquidity (optimal LP sizing) $\rightarrow$ N$\times$Victims $\rightarrow$ 1$\times$Remove Liquidity \\

        \bottomrule
    \end{tabularx}
\end{table*}

We aim to answer the following research questions:
\begin{itemize}
\item\textbf{RQ1: Strategy reproduction and variant exploration.} 
Can \sysname reproduce known strategies and discover new ones?
(Sec.~\ref{Sec_Experiments: Variants})

\item\textbf{RQ2: Cross-protocol adaptation}. 
Can \sysname adapt strategies across heterogeneous AMM protocols while
preserving reusable strategy logic?
(Sec.~\ref{sec:rq2-protocol})

\item\textbf{RQ3: Cross-chain adaptation}. 
Can \sysname migrate strategies to new chains while maintaining execution
quality, profitability, and reasonable costs?
(Sec.~\ref{sec:rq3-chain})

\eat{\item\textbf{RQ4: Runtime performance}. How efficiently do the generated MEV
bots plan and validate opportunities across protocols and chains?}
\item\textbf{RQ4: Planning latency}. How efficiently do the generated MEV
bots construct execution plans across protocols and chains?
(Sec.~\ref{sec:rq4-runtime})

\item\textbf{RQ5: Ablation study}. How does the harness design of \sysname
contribute to its overall capabilities? (Sec.~\ref{sec:rq5-ablation})
\end{itemize}

\stitle{Runtime environment}.
Our experiments run on an Ubuntu 24.04.4 LTS server with Intel Core i9-14900 (24 physical cores / 32 threads, up to 5.8 GHz), 62 GiB RAM, and a 1 TB NVMe volume. We use Foundry/Anvil v1.5.1 for forked EVM execution, Web3.py v7.14.1 for blockchain interaction, SQLite v3.45.1 for artifact storage, and Reth v1.10.2 as the local Ethereum archive node. Agent model and reasoning effort are fixed as GPT-5.4 with xhigh reasoning effort.

\stitle{Datasets \& chain selection.} We comprehensively evaluate \sysname on Ethereum, Base, and BNB Smart Chain (BSC)~\cite{binance_academy_mev_2026, zhang2025following, qin2023imitation}. For each chain, our dataset contains over 1.5 million blocks spanning from Feb. to Aug. 2025. These chains represent three distinct environments in the modern EVM ecosystem. Ethereum serves as the primary baseline, as it remains the dominant venue for decentralized finance (DeFi) liquidity and represents the most complex MEV landscape. BSC is selected as a prominent alternative Layer-1 characterized by high transaction volumes, retail-heavy DeFi activity, and shorter block times, offering a different network state compared with Ethereum mainnet. Finally, Base is selected to represent emerging Layer-2 optimistic rollups. Layer-2 networks typically rely on a centralized sequencer with a private mempool to order transactions and construct blocks \cite{ferreira2024rolling}. This design hides pending transactions from the public network, removing the visibility required by sandwich attacks.

\stitle{Assumptions}. 
\eat{
Note that our experiment isolates the profitability of each strategy by executing them independently and does not model competition with other MEV bots. 
}
To fairly compare MEV variants and their performance across AMMs and chains, we replay each historical transaction independently, assuming the attacker possesses unbounded capital and not considering submission via private relay, excluding competition.
Thus, the revenues stand for the upper bounds. We value profit directly in WETH/WBNB and deduct gas. We do not simulate asset acquisition, on-chain conversion, flash-loan repayment, or builder bribes.

\subsection{RQ1: Strategy Reproduction and Variant Exploration}
\label{Sec_Experiments: Variants}

We first evaluate whether \sysname can reproduce established MEV strategies and
generate novel variants on Ethereum.
We count a strategy as successfully reproduced or generated only when its
planner and execution logic complete under the validator and produce a
positive-profit case. 

\stitle{Reference strategies.}
As shown in Table~\ref{tab:variants},
\sysname successfully reproduced and validated six reference strategies already reported in prior research, spanning
sandwich, backrun, and liquidity-based MEV.
The three basic reference strategies are the canonical single-victim
strategies HFT, SBA, and JIT. HFT executes an optimally sized front-run before
the victim and a backrun afterward~\cite{zhou2021high}; SBA executes a
same-pool backrun after the victim; and JIT adds liquidity before the victim
swap and removes it afterward~\cite{zhou2021just}.
The three complex reference strategies are BUR, LR, and MLL~\cite{sandwichattack}, which
involve either multiple victims or composite attacker actions.
For example, the BUR description
$1\times\text{Frontrun}\rightarrow N\times\text{Victims}\rightarrow
1\times\text{Backrun}$ denotes one attacker transaction before a group of victims and
one attacker transaction after it. To reproduce a reference strategy, technical papers are given to the orchestrator agent. The reproduced codes are further censored by human experts. 

\stitle{Variants.}
Under the MEV variant explorer mode, \sysname generated and validated five novel
variants from reproduced reference strategies. Two patterns emerge: \emph{(1) Larger parameter spaces}:
rather than introducing a new action sequence, the variants replace a fixed or
restricted parameter choice with a broader search over the corresponding
attacker input or liquidity size; 
\emph{(2) HFT composition}:
each variant combines HFT's search-based sizing with the execution structure of another reference strategy. For instance, SBA-HFT and JIT-HFT combine it with basic SBA and JIT, while BUR-HFT, LR-HFT, and MLL-HFT combine it with their complex
counterparts. Consequently, each variant preserves the action ordering and complexity classification of its parent strategy.
The discovered variants yield a total revenue increase of $1.02\times$ to $15.97\times$ across all protocols compared to the evaluated baseline strategies.
We next illustrate one discovered variant.


\begin{figure}[t]
    \centering
    \includegraphics[width=1\linewidth, trim=8.5cm 5cm 5.8cm 2.2cm, clip=true]{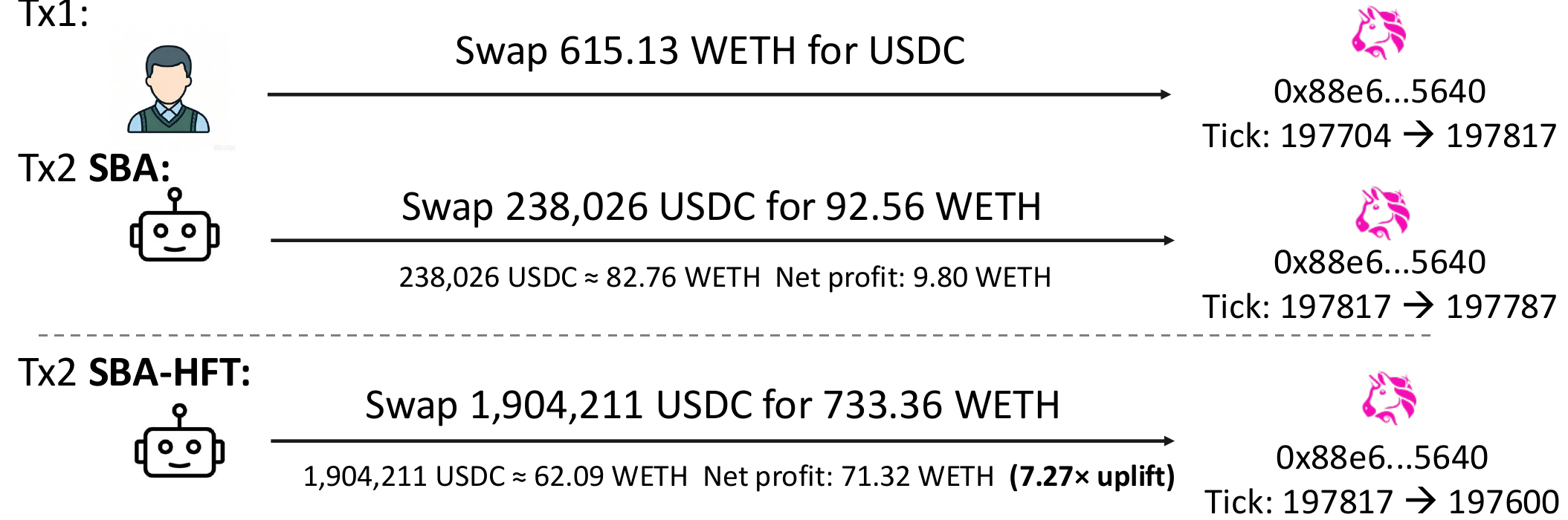}
    \caption{
    SBA vs. SBA-HFT strategy on the same Uniswap V3 opportunity. SBA-HFT yields $7.27\times$ the net profit of SBA.
    }
    \label{fig:SBA-HFT-Case}
\end{figure}

\stitle{Case study: SBA-HFT in a Uniswap V3 pool.}
Fig.~\ref{fig:SBA-HFT-Case} illustrates the execution differences between SBA and SBA-HFT on an Ethereum opportunity (block~\href{https://etherscan.io/block/21764567}{21764567}). In this scenario, a victim transaction~\href{https://etherscan.io/tx/0x249653245220edb78719fdb0924575a78b826bc3be37932705c6a9f7d312c5c5}{\texttt{0x24...c5c5}} swaps $615.13$ WETH for USDC, routing $369.02$ WETH through the focal Uniswap V3 $0.05\%$ USDC/WETH pool~\href{https://etherscan.io/address/0x88e6a0c2ddd26feeb64f039a2c41296fcb3f5640}{\texttt{0x88...5640}}. This massive trade creates a significant price displacement, moving the active tick from 197704 to 197817. Following this victim swap, an attacker can backrun the transaction by executing a reverse trade (swapping USDC back for WETH) to pull the pool price toward its original state. Both strategies operate on the exact same post-victim state but differ fundamentally in how they determine the size of this reverse trade.

\etitle{Result.}
While both strategies successfully extract arbitrage, their differing sizing mechanisms lead to a massive discrepancy in returns. SBA selects a backrun input of 238,026 USDC, yielding a net profit of 9.80 WETH (0.004 WETH in gas). Conversely, SBA-HFT inputs a significantly larger 1,904,211 USDC order, recovering 733.36 WETH. Valuing the USDC inputs at the daily rate ( $3.477\times10^{-4}$ WETH/USDC) and subtracting a 0.010 WETH gas cost, SBA-HFT achieves a net profit of 71.32 WETH. Thus SBA-HFT captures $7.27\times$ the net profit of standard SBA.

\etitle{Analysis \& Insights.}
The profound performance gap stems from how each strategy explores the available liquidity. Standard SBA calculates a base trade size from the victim's volume and current pool liquidity, then tests a fixed list of multipliers (e.g., 0.5x, 1x, 1.5x) to select its final backrun amount. Its conservative sizing stops the backrun prematurely at tick 197787, leaving substantial arbitrage unextracted. In contrast, SBA-HFT employs a multi-stage algorithmic search (\textsc{SizeOptimization} in Alg.~\ref{alg:sba_hft}). It first conducts coarse probing across larger orders of magnitude, applies a ternary search to bound the most promising region, and finally executes a granular local scan to pinpoint the optimum $s^*$. 

While the underlying HFT search algorithm pre-existed in a codebase repository, \sysname successfully performed adaptive migration and applied it to a single-pool swap-backrun scenario. This HFT-style optimization allows the strategy to confidently scale its backrun deeper into the pool down to tick 197600. Ultimately, this variant demonstrates the \sysname's capacity to deeply comprehend parameter search algorithms and correctly fuse them into established MEV architectures, replacing heuristics with precise numerical optimization to calculate the exact trade volume that maximizes net profit against the pool's price impact curve.

\subsection{RQ2: Cross-Protocol Adaptation}
\label{sec:rq2-protocol}

\eat{To evaluate \sysname's ability to generalize across heterogeneous AMM designs, we run its protocol adaptation mode to adapt MEV strategies to two target protocols on Ethereum: Uniswap V3, a concentrated-liquidity AMM (CLMM), and Balancer V2, a weighted-pool AMM.}
We next evaluate whether \sysname can adapt MEV strategies across distinct AMMs. Starting from CPMM implementations, \sysname adapts $11$ strategies to two target protocols on Ethereum: Uniswap V3, a concentrated-liquidity AMM (CLMM), and Balancer V2, a weighted-pool AMM. We measure profitability, opportunity conversion, and code reuse.


\begin{table*}[!htb]
    \centering
    \caption{Performance of MEV strategies on AMMs and chains. \textbf{Plan./Cand.}, \textbf{Sim./Plan.}, \textbf{Pro./Sim.}, and \textbf{Pro./Cand.} denote the ratios of planned cases to candidates, successfully simulated cases to planned cases, profitable value cases to simulated cases, and profitable value cases to candidates, respectively (Definitions of four case categories in App.~\ref{four_case_categories}). Absolute counts are retained for candidates, profitable cases and total revenue.}
    
    \label{tab:opportunity_conversion}
    
   \scriptsize
    \begin{tabular}{c l l l l c c c c l l}
        \toprule
        \textbf{Chain} & \textbf{Strategy} & \textbf{Protocol} & \textbf{Candidates} & \textbf{Profited}
        & \textbf{Plan./Cand.} & \textbf{Sim./Plan.} & \textbf{Pro./Sim.} & \textbf{Pro./Cand.}
        & \textbf{Revenue (ETH or BNB)} \\
        \midrule
        \multirow{33}{*}{Ethereum} 
        &\multirow{3}{*}{HFT} 
        & CPMM & 1,116,645  & 26,394 & 43.31\% &18.48\% &  29.53\% & 2.41\% & 1,072.06 \\
        
        && CLMM &18,656,365 &57,659  &14.00\%  & 7.22\%   & 30.56\% & 0.31\%&3,920.81 \\
        
        & & Balancer &321,207  &598  &7.33\%     & 7.08\%     & 35.83\% & 0.19\%& 27.36 \\

        \cmidrule{2-10}
        &\multirow{3}{*}{SBA} 
        &CPMM & 4,470,901   &191,453 &66.00\%  &25.41\%   &25.53\% &4.28\%  & 754.08 \\
        && CLMM &18,645,815 &186,171 &27.47\%& 33.15\%   &10.96\% & 1.00\% & 3,837.03 \\
        &&Balancer &245,724 &12,666  &26.59\%    &75.98\%   &25.52\% &5.15\%  & 223.12 \\

        \cmidrule{2-10}
        &\multirow{3}{*}{SBA-HFT} 
        &CPMM &4,535,461 &26,579 & 72.62\% & 25.16\%  & 3.21\% & 0.59\% & 21,475.46 \\

        && CLMM &18,395,094 &184,866 & 35.31\%  & 35.71\%  & 7.70\% &  1.00\% &17,972.26  \\
        
        && Balancer & 245,724 & 6,980  & 7.29\% & 75.09\%  & 51.90\% & 2.84\% &201.85 \\
        
        \cmidrule{2-10}
        &\multirow{3}{*}{JIT} 
        
        & CPMM &4,669,731  & 26,948   &22.71\%    &7.71\%    &32.94\% & 0.58\% & 317.10\\
        && CLMM &11,508,279 & 1,719,907 &69.14\% & 68.41\% & 31.60\% &14.94\%& 49,665.13\\
        &&Balancer & 321,207  &882  &93.11\%  &73.41\%     & 0.40\% &0.27\%  & 4.23 \\

        \cmidrule{2-10}
        &\multirow{3}{*}{JIT-HFT} 
        &CPMM &26,972,006 &169,599 & 21.92\% & 8.08\% & 35.49\% & 0.63\% &1,715.75  \\
        && CLMM &  15,753,115 &1,874,402 & 57.45\% & 61.84\% & 33.49\% & 11.90\% & 49,297.76 \\
        && Balancer & 320,905 & 3994  & 97.17\% & 73.58\% &   1.74\% & 1.24\% & 54.74 \\

        \cmidrule{2-10}
        &\multirow{3}{*}{BUR} 

        &CPMM & 395,874    &40,159 &84.33\%   &55.14\% & 21.82\% & 10.14\% &663.63 \\
        &&CLMM &1,498,086  &30,718  & 40.00\%  &21.13\% & 24.26\% & 2.05\%  & 244.71 \\
        &&Balancer &1,427 &100      &40.78\%   &43.13\%   & 39.84\% & 7.01\%  & 3.56    \\

        \cmidrule{2-10}
        &\multirow{3}{*}{BUR-HFT} 
        & CPMM &398,658 & 41,989 & 84.43\% & 34.99\%  & 35.65\% & 10.53\% & 14,309.44\\
        &&CLMM &1,492,079 &30,456 & 40.32\% & 21.18\% & 23.89\% & 2.04\% & 245.31 \\
        &&Balancer &1,427 &6 & 38.61\% & 1.45\%  & 75.00\% & 0.42\% &5.44  \\

        \cmidrule{2-10}
        &\multirow{3}{*}{LR} 
        &CPMM & 4,954,331 & 15,681   &16.58\%  &15.10\% & 12.64\% & 0.32\%& 53.21\\
        && CLMM &16,951,928 & 8,738 & 0.34\% & 21.39\%  & 70.19\% & 0.05\% & 825.15 \\

        &&Balancer & 321,207 &4,841 & 99.61\% & 61.78\%  & 2.45\% & 1.51\% & 28.11\\ 
        
        \cmidrule{2-10}
        &\multirow{3}{*}{LR-HFT} 
        &CPMM & 1,977,364 &15,579 & 15.89\%  &  15.25\% & 32.87\% & 0.79\% & 1,459.36  \\
         && CLMM & 18,852,384 &  5,647  & 0.07\%  & 57.78\% & 72.14\% &  0.02\% &   677.63 \\
        
        && Balancer & 321,207 & 3764  & 99.61\% & 59.85\%& 1.96\% & 1.17\% & 46.02  \\

        \cmidrule{2-10}
        &\multirow{3}{*}{MLL} 
        & CPMM &658,553 & 11,299 & 9.06\%  & 30.26\%  & 62.58\% & 1.72\%& 761.48\\
        
        &&CLMM &1,308,375 &107,775& 91.38\% & 57.41\% & 15.70\% & 8.24\% &2451.15 \\
        
        && Balancer &10,895 &1,085 & 100\%  & 29.45\%  & 33.81\% & 9.96\% & 34.77 \\
        
        \cmidrule{2-10}
        &\multirow{3}{*}{MLL-HFT} 
        &CPMM &658,553  &11,736  & 9.10\%  & 31.01\%& 63.15\%  &  1.78\% & 496.42 \\
        
        &&CLMM &936,470  &122,210 & 87.98\% & 66.54\% & 22.29\% & 13.05\%  & 24,210.40 \\

        && Balancer &41,657 & 4,885  & 97.19\%  & 31.21\% & 38.67\% & 11.73\% & 90.16\\

        \midrule
        \multirow{6}{*}{BSC} 
        
        & HFT &CPMM&5,099,610 & 30,127 & 29.97\% & 29.55\%  & 6.67\% & 0.59\% &5,307.04\\

        &SBA &CPMM&3,380,574 & 308,939 & 21.87\% & 63.99\% &  65.28\% & 9.14\% & 2,283.78 \\

        &JIT &CPMM& 246,306& 31,548  &92.00\%  &67.90\%&  20.50\% & 12.81\%& 420.83 \\
        
        & BUR &CPMM&322,471 & 58,833 & 46.27\%  & 41.25\%  & 95.60\% & 18.25\% & 1,940.28 \\
        
        & LR &CPMM & 537,843 & 25,981 & 45.66\% & 25.22\%  & 41.94\% & 4.83\% & 375.72 \\
        
        & MLL &CPMM&226,299 & 77,030 & 50.53\%  & 68.82\% &  97.88\% & 34.04\% & 2327.51 \\
        \midrule

        \multirow{2}{*}{Base} 
        & SBA &CPMM& 195,040 & 2,260& 27.43\%  & 89.31\%  & 4.73\% & 1.16\% & 16.13 \\

        & SBA-HFT &CPMM& 823,099 &3944   &36.57\%  & 90.19\% & 1.45\% & 0.48\% &1682.92 \\

        \bottomrule
    \end{tabular}
\end{table*}

\stitle{Profitability.} 
Table~\ref{tab:opportunity_conversion} reports revenue and opportunity
conversion. \sysname produces executable and profitable adaptations across both target AMMs. Balancer consistently yields the fewest candidates and lowest revenue, due to the substantially lower trading activity on it. We next highlight several findings.


\etitle{JIT-HFT.} 
\emph{(1) High CLMM revenue.} JIT-HFT and JIT generate the highest total revenue on CLMMs compared to all other strategies, driven by three factors. First, the high volume of DeFi activity on CLMMs provides large candidate sets, while JIT mechanics make a broad range of swaps susceptible to fee-capture opportunities. Second, our evaluation assumes unbounded attacker capital. Because executing JIT on CLMMs may require liquidity hundreds of times larger than the victim’s swap volume~\cite{xiong2023demystifying}, this assumption removes a major practical bottleneck. Third, adding JIT liquidity generally does not violate the victim’s slippage constraint, allowing the victim transaction to execute and the attacker to capture its trading fees.
\emph{(2) Profit distribution.} JIT-HFT profits are concentrated in rare, capital-intensive cases (Tab.~\ref{tab:jit-hft-v3-profit-dist} in Appendix). Although 84.4\% of profitable cases earn less than 0.01\,WETH each, they contribute only 5.35\% of total profit. In contrast, the 15.6\% of cases earning 0.01--10\,WETH contribute 88.5\% of total profit. 
Larger profits also require substantially more capital: average capital rises from 1.25\,WETH for cases earning below 0.0001\,WETH to 4{,}259\,WETH for those earning at least 10\,WETH.
\emph{(3) Limited benefit over JIT.}
While HFT-based sizing increases JIT revenue by 5.41$\times$ on CPMM, it provides no revenue gain on CLMM, since the heuristic sizing used by JIT already captures nearly all available CLMM revenue and leaves little room for further optimization.

\etitle{BUR-HFT.} BUR-HFT increases BUR’s revenue by 21.56$\times$ on CPMM but yields negligible gains on CLMM. This discrepancy stems from the structural divergence between the two liquidity models. BUR-HFT replaces fixed front-run sizing with a continuous-style ternary search algorithm, which assumes a profit curve with a single peak. This approach works well with CPMM's continuous pricing curve. However, CLMM distributes liquidity discretely across specific price ticks, resulting in a partitioned execution surface. On CLMM, even if profit changes predictably within adjacent ticks, the overall profit function does not necessarily follow a single-peak pattern. Consequently, the ternary search fails to find the global optimum, revealing that the agents cannot yet tailor their optimization methods to more complex liquidity models.

\eat{
\etitle{BUR-HFT.}
BUR-HFT increases BUR’s revenue by 21.58$\times$ on CPMM but yields negligible gains on CLMM. The discrepancy stems from the structural divergence between the two liquidity models. BUR-HFT replaces fixed front-run sizing with search algorithms, which work well with CPMM's continuous pricing curve. However, CLMM distributes liquidity discretely across ticks, resulting in a discontinuous profit function that limits the benefit of precise sizing optimization. This comparison reveals that the agents cannot yet tailor their optimization methods to more complex liquidity models.
The lack of a large BUR-HFT gain on CLMM is that the ternary search wrongly assumes the profit curve rises to one single peak, treating a continuously rising trend as a simple peak problem. The planner performs a continuous-style ternary search over a CLMM execution surface that is partitioned by tick crossings and liquidity changes. However, on CLMM, even if the profit changes within adjacent initialized ticks, when considering all intervals together, the overall profit does not necessarily follow a pattern of one single peak.
}


\etitle{MLL-HFT.} MLL-HFT earns less total revenue than MLL on CPMM. Across their 11,244 shared profitable cases, 15.50\% produce equal revenue. Among the remaining cases, MLL outperforms in 67.37\%, whereas MLL-HFT performs better in 32.63\%.
This underperformance occurs because MLL-HFT inappropriately applies the same ternary search algorithm to a non-peaking curve. In CPMM, the profit for liquidity provision generally increases with the amount provided and then flattens. Consequently, the baseline MLL produces the best outcome simply by selecting the maximum allowed limit ($20\times$ the victim trading volume). When MLL-HFT applies ternary search to this flattening curve, the profit differences between the sampled points at larger liquidity sizes become indistinguishable, causing the search to converge prematurely before reaching the optimal upper limit. 
\eat{
This difference arises because the ternary search used by MLL-HFT does not match the underlying CPMM model.
In CPMM, profit generally increases with the amount of liquidity provided and then flattens. The baseline MLL produces the best outcome by simply choosing the maximum allowed liquidity limit, which equals $20\times$ the victim trading volume. In contrast, MLL-HFT applies ternary search that wrongly assumes the profit curve rises to one single peak, treating a continuously rising trend as a simple peak problem. When the objective becomes nearly flat at larger liquidity sizes, its probes become indistinguishable and the search may converge before reaching the upper limit.
}

\stitle{Opportunity conversion.}
The conversion rates identify where the protocol adaptation fails. Plan./Cand. measures whether \sysname can instantiate the strategy for a candidate on the target AMM. Sim./Plan. measures whether the adapted plan executes under the AMM's protocol semantics. Pro./Sim. measures whether an executable adaptation remains profitable after deducting fees and gas, and whether its profit can be valued. Appendix~\ref{four_case_categories} gives detailed definitions.

For BUR-HFT on Balancer V2, Sim./Plan. is only 1.45\%, compared with 43.13\% for BUR. Aggressive front-run sizing reverts the victim transaction in 92.4\% of failed cases, leaving six profitable ones.
Balancer's asymmetric weighted pricing makes approximate sizing overestimate the safe front-run amount and violate victim slippage. The low conversion rate thus indicates a protocol-specific sizing mismatch, rather than absent candidates.
For SBA on CLMMs, the planning and validator stages yield 1.69 million simulations, of which only 10.9\% were profitable due to transaction costs: 75.3\% of the simulated cases fail to be profitable due to pool fees, 4.1\% are unprofitable by gas costs, and another 9.7\% are excluded due to the lack of daily normalized prices for certain tokens.




\begin{table}[t]
\centering
\small
\caption{Code Similarities Across Protocol Adaptations (C=CPMM, CL=CLMM, B=Balancer V2).}
\label{tab:similarity_metrics}
\begin{tabularx}{\columnwidth}{l X X X X}
\toprule
Strategy & JSCPD & JSCPD & JPlag  & JPlag  \\
        &(C--CL) &(C--B) &(C--CL) &(C--B) \\
\midrule
HFT     & 0.53\% & 1.25\% & 9.49\%  & 15.20\% \\
SBA     & 0.00\% & 0.84\% & 12.53\% & 17.04\% \\
JIT     & 0.86\% & 0.78\% & 23.51\% & 17.33\% \\
SBA-HFT & 0.00\% & 0.00\% & 20.08\% & 23.26\% \\
JIT-HFT & 1.90\% & 1.43\% & 22.01\% & 20.87\% \\

\midrule
Average & 0.66\% &0.86\%  & 17.52\% &  18.74\% \\
\bottomrule
\end{tabularx}
\end{table}

\stitle{Code reuse.}
\eat{
We measure whether \sysname discovers reusable strategy logic across AMMs by
comparing five CPMM implementations with their CLMM and Balancer V2
adaptations. JSCPD~\cite{jscpd} measures exact text reuse, while
JPlag~\cite{prechelt2002finding} measures token-level structural similarity.
}
We measure code reused by \sysname during protocol adaptation by
comparing five CPMM implementations with their corresponding CLMM and Balancer V2 adaptations. JSCPD~\cite{jscpd} measures exact text reuse, while JPlag~\cite{prechelt2002finding} is capable of detecting pairwise similarities even in the presence of code obfuscation.
As shwon in Tab.~\ref{tab:similarity_metrics},
exact reuse averages only 0.66\% for CLMM and 0.86\% for Balancer V2, because pool state, quoting, calldata, and liquidity operations differ significantly across AMMs. Structural similarity reaches 17.52\% and 18.74\%, respectively. The matching regions, detailed in Tab.~\ref{tab:matching_regions}, occur primarily in the planner. Overall, this code reuse analysis demonstrates that the planner is consistently reused during cross-protocol adaptation, intelligently isolating necessary adaptation efforts to the collector and validator.

\subsection{RQ3: Cross-Chain Adaptation}
\label{sec:rq3-chain}
Using CPMM as the target protocol, we next run \sysname to migrate MEV strategies from Ethereum to Base and BSC. Because Base’s private mempool prevents observation of pending transactions, we evaluate only the two backrunning strategies on Base. On BSC, we evaluate all six strategies: HFT, SBA, JIT, BUR, LR, and MLL.

\stitle{Cross-chain performance.}
Table~\ref{tab:opportunity_conversion} summarizes the resulting profitable
cases and revenue. On Base, SBA-HFT earns 1{,}682.92 ETH compared with 16.13 ETH for SBA. On BSC, the migrated strategies also produce profitable cases, including 5{,}307.04 BNB for HFT, 2{,}283.78 BNB for SBA, and 420.83 BNB for JIT. These results show that the adapted strategies remain executable and profitable under different chain environments, although the absolute revenue is not directly comparable across chains or strategy opportunity sets.


\begin{figure*}[t]
    \centering
    \includegraphics[width=1\textwidth]{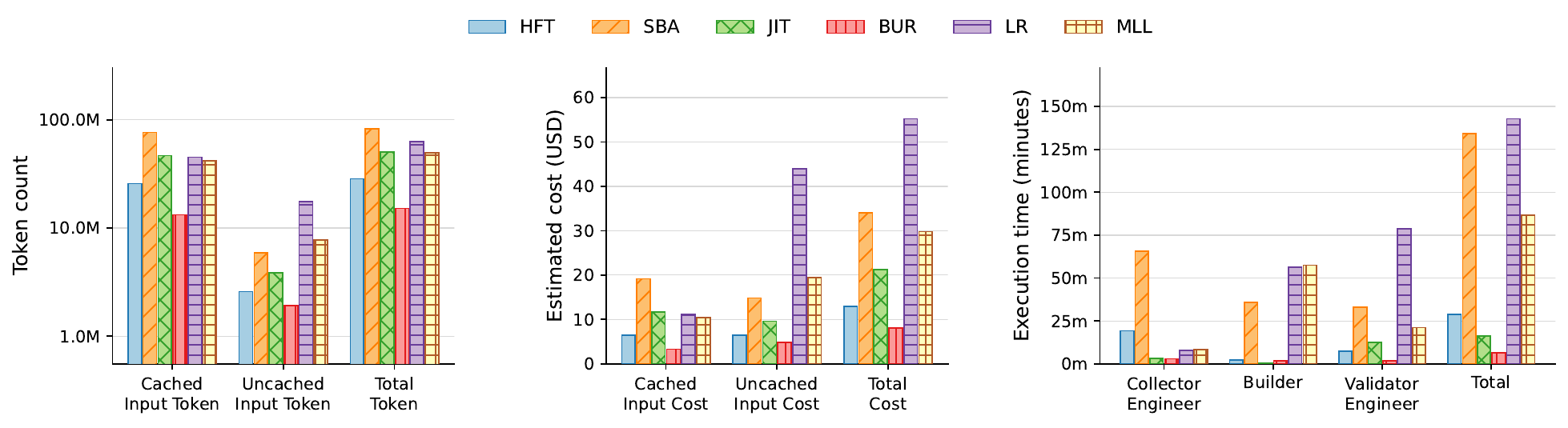}
    \vspace{-2em}
    \caption{Tokens and time consumed during cross-chain migration on BSC. Left: tokens consumption (log scale), grouped by cached input tokens, uncached input tokens, and total tokens. Middle: estimated API cost (USD) under gpt-5.4 pricing~\cite{openai2026gpt54}. Right: per-role agent execution time (collector engineer, builder, validator engineer) and the total agent execution time.}
    \vspace{-1em}
\label{fig:token-usage}
\end{figure*}

\stitle{Economic feasibility}.
To characterize the resource cost of cross-chain adaptation, we evaluate the tokens and time consumed when adapting six reference strategies to CPMM on BSC. The adaptation is executed in two sequential batches, starting with (HFT, SBA, JIT) and followed by (BUR, LR, MLL), with those in each batch processed simultaneously. Figure~\ref{fig:token-usage} reports the cached, uncached, and total input tokens, together with the per-agent execution time for the migration. Cached input tokens are prompt tokens that the model retains in memory and reuses across consecutive requests, which reduces redundant computation.  Total input tokens aggregate the usage by the orchestrator agent and all unique subagent sessions, while uncached input is the difference between total and cached input. SBA consumes the most total tokens (83M), LR has the longest runtime (2.4h) and the highest uncached-input cost, while BUR consumes the lowest. 

BUR migration completes substantially faster than the earlier ones because the agent reuses the BSC infrastructure already built and validated during the HFT adaptation, instead of migrating from scratch. The data collection pipeline, fork based runtime environment, and pricing library are all adopted without modification. \sysname only needs to implement the strategy specific components (i.e., the collector and the planner) and a thin adaptation layer that connects them to the existing pipeline, allowing validation to pass on the first attempt without rework. This efficiency is enabled by \sysname's stage decomposition and explicit interfaces, which let the agent inspect prior artifacts and modify only the components that differ. Consequently, BUR migration requires fewer tokens and less execution time, suggesting that the marginal cost of adapting additional strategies decreases as shared chain infrastructure accumulates.

\begin{table}[t]
\centering
\small
\setlength{\tabcolsep}{4pt}
\caption{P50 and P90 planning latency across protocols and chains
(CP=CPMM, CL=CLMM, B=Balancer V2).}
\label{tab:planning-latency}
\begin{tabularx}{\columnwidth}{@{}ll*{2}{>{\centering\arraybackslash}X}@{}}
\toprule
Chain & Strategy & P50 (ms) & P90 (ms) \\
\midrule
\multirow{15}{*}{Ethereum}
 & HFT--CP & 0.20 & 0.24 \\
 & HFT--CL & 17 & 35 \\
 & HFT--B & 115 & 212 \\
 & SBA--CP & 0.16 & 0.21 \\
 & SBA--CL & 10 & 75 \\
 & SBA--B & 6.8 & 12 \\
 & JIT--CP & 4.7 & 6.3 \\
 & JIT--CL & 3.3 & 9.6 \\
 & JIT--B & 5.4 & 10 \\
 & SBA-HFT--CP & 4.4 & 6.9 \\
 & SBA-HFT--CL & 146 & 658 \\
 & SBA-HFT--B & 129 & 216 \\
 & JIT-HFT--CP & 68 & 103 \\
 & JIT-HFT--CL & 5.1 & 21 \\
 & JIT-HFT--B & 119 & 285 \\
\midrule
\multirow{2}{*}{Base}
 & SBA--CP & 0.31 & 0.36 \\
 & SBA-HFT--CP & 11& 23 \\
\midrule
\multirow{6}{*}{BSC}
 & BUR--CP & 0.33 & 0.98 \\
 & HFT--CP & 2.0 & 4.9 \\
 & JIT--CP & 0.51 & 1.0 \\
 & LR--CP & 0.16 & 0.63 \\
 & MLL--CP & 1.1 & 2.8 \\
 & SBA--CP & 0.22 & 0.65 \\
\bottomrule
\end{tabularx}
\end{table}

\subsection{RQ4: Planning Latency}
\label{sec:rq4-runtime}

An MEV opportunity can disappear once another transaction changes the relevant state. We therefore measure planning latency, defined as the interval between the planner receiving candidates and generating a complete execution plan.
As shown in Table~\ref{tab:planning-latency}, while most Ethereum CPMM and all BSC strategies complete planning within 10\,ms at P90, SBA-HFT on CLMM requires 658\,ms.
The latency of the SBA-HFT planner originates from a two-layer nested computation to determine the optimal backrun amount. The outer layer performs a ternary search before sequentially scanning candidates, thereby generating numerous amounts to evaluate. The inner layer then replays each candidate by simulating the victim swap to establish the updated price, followed by the reverse backrun to calculate profit. Crucially, CLMM makes this replay an intrinsically discrete process because liquidity exists exclusively between initialized ticks. Advancing a swap requires first computing the amount consumed within the current active liquidity. Upon crossing a tick, the planner must update the active liquidity before continuing the computation for the remaining amount. Since large victim and backrun amounts inevitably traverse multiple ticks, these traversals proportionally increase the computational load.

\stitle{Feasibility on real-time blockchains}.
On Ethereum, the average block interval is 12\,s, and a victim transaction spends 23--35\,s in the mempool~\cite{zhou2021high}. In contrast, the maximum P90 planning latency for HFT is only 212\,ms, easily allowing for real-time execution. Base introduces ``Flashblocks'' to provide pre-confirmations within 200\,ms~\cite{wang_flashblocks_2025}, while its maximum P90 planning latency is 23\,ms. Similarly, while the fast finality on BSC reduces confirmation times to under 2\,s, the maximum P90 planning latency is merely 4.9\,ms, significantly shorter than the opportunity window. Note that end-to-end latency also depends on infrastructure factors beyond our scope, such as network propagation and transaction submission, which can be reduced through deployment optimizations.


\begin{table}[h]
\centering
\caption{Ablation results on BSC blocks 46262540--46272539 for \sysname and pipeline-structured agent (PSA).}
\label{tab:ablation}
\vspace{0.5em}
\setlength{\tabcolsep}{4pt} 
\small 
\resizebox{\columnwidth}{!}{%
\begin{tabular}{llrrrrr}
\toprule
System & Strategy & Candidates & Planned & Runtime & Profit & Total Profit \\
       &          &            &         & Success & Found  & (WBNB) \\
\midrule
(a) \sysname
& HFT   & 20,890 & 5,222  & 3,307  & 2,002 & 23.80 \\
& SBA   & 19,777 & 6,949  & 4,342  & 2,480 & 7.17  \\
& JIT   & 17,690 & 17,683 & 12,282 & 64    & 0.19  \\
\cmidrule{2-7}
& Total & 58,357 & 29,854 & 19,931 & 4,546 & 31.17 \\
\midrule
(b) PSA
& HFT   & 9,953  & 708    & 435    & 28  & 0.09 \\
& SBA   & 12,444 & 2,329  & 2,108  & 53  & 3.73 \\
& JIT   & 12,442 & 12,442 & 9,158  & 28  & $1.09 \times 10^{-7}$ \\
\cmidrule{2-7}
& Total & 34,839 & 15,479 & 11,701 & 109 & 3.81 \\
\bottomrule
\end{tabular}%
}
\end{table}



\subsection{RQ5: Ablation Study}
\label{sec:rq5-ablation}

We compare \sysname with a single-agent baseline, denoted \emph{Pipeline-Structured Agent} (PSA) on three BSC CPMM strategies: HFT, SBA, and JIT.
Both use the same task objectives, reference strategies, execution modules, and profitability evaluator. PSA is given the three-stage decomposition (i.e., data collector, planner, and validator) and the corresponding code regions, and may iterate between these stages. Unlike \sysname, PSA does not employ specialized subagents, explicit responsibility boundaries, or diagnostic loops. This ablation therefore isolates the benefits of role-specialized multi-agent coordination. Tab.~\ref{tab:ablation} summarizes the results on 10,000 BSC blocks.


\stitle{Findings.}
The performance gap between \sysname and PSA arises primarily due to the smaller set of candidate and planned cases produced at the collector and planner stages, rather than from simulation failures. At the collector stage, PSA's HFT accepts only three CPMM router selectors, whereas \sysname's collector supports six, including fee-on-transfer paths. Consequently, PSA silently drops a subset of legitimate transactions. Moreover, PSA fetches transactions and receipts individually, whereas \sysname retrieves each entire block together with all of its transaction receipts.
The block-level fetching reduces RPC failures and preserves complete intra-block execution context, allowing \sysname to filter more accurately and avoid missing valid opportunities.
As a result, \sysname identifies more valid opportunities. Across HFT, SBA, and JIT, \sysname consistently finds roughly 17,690 to 20,890 candidates, compared to PSA’s 9,953 to 12,444.

At the planner stage, the gap widens further. 
PSA's BSC HFT planner is essentially a thin wrapper around the reference Ethereum HFT planner. It adjusts only basic chain parameters and addresses, while relying on standard off-the-shelf math to calculate trades, showing a lack of customized logic. Consequently, the key difference is how the planners bound the front-run search range. The PSA planner starts from a small probe value and expands exponentially without a tight upper bound, allowing the search to reach unrealistically large values and produce front-run amounts that are often infeasible on chain. In contrast, \sysname's BSC planner explicitly caps front-run sizes at 1\% of the pool reserve. This conservative bound ensures better on-chain executability and yields more execution plans. As a result, the planned rate for PSA's HFT is only 7.11\%, compared with 25.00\% for \sysname. 

Similarly, the planners of SBA and JIT are also lightweight wrappers around the Ethereum reference planners. The reused planner oversimplifies the input data structures, pool snapshots, and victim constraints, causing many BSC-specific fields to be lost or improperly generated when assembling the execution plan. Candidates are discarded during simulation because the required data are incomplete. 
This pattern reveals the limitation of PSA. A single agent controlling the entire MEV bot struggles to maintain fine-grained control over every stage's fields. By contrast, a multi-agent architecture can assign different specialist agents to optimize different stages respectively, and repeated revision loops can continuously strengthen this during fine-grained code repair.
\section{Applications of \sysname}
\label{sec:applications}


We present three distinct applications to demonstrate the utility and broader
implications of our approach. 

\etitle{Democratizing MEV}.
Rather than eliminating MEV extraction, the Ethereum community has widely adopted auction platforms to facilitate efficient, decentralized and transparent MEV extractions~\cite{yang2023sokmevcountermeasurestheory}. 
Unfortunately, this mechanism simultaneously caused technical monopolies among a few institutions. A small number of entities capture the majority market share relying on complex implementation capabilities. By autonomously generating MEV Bots, \sysname significantly lowers the technical barrier to entry, empowering smaller searchers to compete, allowing Decentralized Autonomous Organizations (DAOs) to capture MEV profits directly for their own treasuries rather than losing them to third parties, and ultimately enhancing the decentralization of blockchains.
\etitle{Accelerating cross-protocol and cross-chain adaptation}.
As the blockchain ecosystem expands into a multi-protocol and multi-chain paradigm, relying on manual scripting introduces severe engineering burden. The initial deployment of new Layer-2 rollups or novel DeFi protocols typically creates a highly profitable, low-competition ``vacuum period.'' However, because traditional extraction requires engineers to manually engineer on novel smart contracts and rebuild logic from scratch, they often fail to deploy quickly enough to capitalize on these fleeting first-mover advantages. 

\sysname fundamentally resolves this engineering bottleneck by decoupling strategy formulation from hardcoded protocol specifics. When migrating to a new blockchain or integrating an unfamiliar AMM, \sysname autonomously parses the target protocol's ABIs, documentation, and state variables to dynamically reconstruct the necessary execution payloads. This allows the system to seamlessly port existing MEV strategies to emerging blockchains or adapt to novel AMM mechanisms without human intervention. Consequently, \sysname enables extractors to rapidly scale their operations across disparate ecosystems and capture value during critical vacuum periods that are otherwise constrained by human capability.


\etitle{Promoting the intersection of AI and DeFi}.
The integration of LLMs and autonomous agents into DeFi remains in its nascent stages. While AI has demonstrated proficiency in code generation and static text analysis \cite{gervais2026aiagentsmartcontract}, its capability to navigate real-time, adversarial financial environments, such as the MEV dark forest, lacks systematic evaluation. 

By deploying \sysname as a pioneering framework for autonomous MEV strategy generation, we systematically evaluate the efficacy and limitations of AI-driven Web3 agents. Our application serves as a benchmark for measuring how well AI models can comprehend complex on-chain states, interpret custom protocol ABIs, and execute rigorous MEV bot workflow without human intervention. Researchers can utilize the logs and operational metrics generated by \sysname, such as its planning latency in strategy formulation, as empirical data to refine future Web3-specific AI models. Furthermore, this application provides actionable insights into how AI will inevitably influence MEV behavior and network dynamics, inspiring future research directions for securely and seamlessly integrating AI into decentralized economies.

%



\section{Related Work}

MEV research has progressed from early studies characterizing MEV exposure~\cite{torres2021frontrunner,qin2022quantifying,zhou2021just} to large-scale empirical analyses~\cite{mclaughlin2023large} and the development of countermeasures~\cite{yang2023sokmevcountermeasurestheory, kavousi_et_al:LIPIcs.OPODIS.2025.36, alpos2026censorship, cryptoeprint:2026/1195}. As Layer-2 networks and rollups have proliferated, recent work has further examined how their distinct sequencing and execution models reshape MEV~\cite{ferreira2024rolling, gogol2025priority}. Alongside these efforts, agentic AI systems have begun automating DeFi security analysis and smart contract vulnerability exploitation. We discuss the two systems most closely related to \sysname.


A1~\cite{gervais2026aiagentsmartcontract} proposes an end-to-end agentic framework that autonomously generates smart contract exploits. It features a closed-loop feedback mechanism, where the AI iteratively refines its attack strategy, written as Solidity code chunks, based on revenue indicators and execution results, demonstrating the ability of AI agents to search complex exploit spaces. However, while this framework excels at finding vulnerabilities within a certain static environment, A1 heavily relies on a fixed set of human-engineered tools for data retrieval, execution, and revenue calculation. Moreover, A1 is explicitly designed for smart contract exploitation rather than for the MEV ecosystem. In contrast, \sysname provides a longer end-to-end agentic chain and stronger versatility. It decomposes the MEV bot workflow into three modular stages and employs specialized agents to autonomously refine their implementations, addressing the limitations of static toolkits. \sysname effectively replaces manual maintenance with agent-driven self-adaptation, enabling the system to autonomously calibrate MEV strategies to non-static environments, including previously unseen protocols and blockchains.

TxRay~\cite{wang2026txray} focuses on postmortem analysis and reconstruction of blockchain attacks. Starting from seed transactions, it autonomously performs analysis, evidence verification, and forked Foundry execution to generate reproducible PoCs. In terms of functionality, however, TxRay is designed to reconstruct observed attacks rather than generate MEV bots or adapt strategies to new environments. It also does not explore new exploitation strategies beyond the observed attack set. In terms of agent harness design, \sysname and TxRay differ significantly in role assignment, task division, and revision loops. Moreover, TxRay uses fixed handoff templates, whereas \sysname supports flexible natural-language handoffs. Overall, prior systems primarily automate the reproduction or exploitation of known attacks, while \sysname moves toward proactive MEV strategy discovery and adaptation.


\section{Conclusion}
To our knowledge, \sysname is the first fully autonomous multi-agent system for MEV bot generation. Empirical experiments on over 1.5 million historical blocks and three different blockchains validate its effectiveness in MEV variant discovery, cross-protocol adaptation, and cross-chain adaptation.



\balance

\bibliographystyle{plain}
\bibliography{main}

\cleardoublepage
\appendix

\section*{Ethical Considerations}
To ensure our study causes zero negative impact on live DeFi ecosystems while upholding rigorous scientific conduct, all empirical evaluations of \sysname are strictly confined to isolated historical replay environments. We structure our ethical assessment across five core pillars as below:

\stitle{Experimental Safeguards \& Zero Harm:}
\sysname involves no real trading or live on-chain execution. All experiments run within isolated local forks spanning 1.5M historical blocks across Ethereum, BSC, and Base. Relying exclusively on public ledger data without touching credentials or broadcasting transactions, our offline analysis causes zero financial loss or disruption to live DeFi ecosystems.

\stitle{Stakeholder Impact:} 
We identify five key affected groups: 
(1)~\textit{Protocol developers and security teams}, primary beneficiaries who gain preemptive insights to defend against newly uncovered MEV variants; 
(2)~\textit{On-chain users and liquidity providers}, long-term beneficiaries of hardened protocols and reduced centralized value extraction; 
(3)~\textit{Academic and industrial security researchers}, who obtain transparent benchmarks for automated MEV analysis; 
(4)~\textit{Legitimate searchers and block builders}, who benefit from lowered barriers to decentralized, open MEV strategy development; and 
(5)~\textit{Malicious adversaries}, a high-risk secondary stakeholder whose exploitation incentives are curtailed through our defensive framing and tiered release scheme.

\stitle{Ethical Principles \& Compliance:} 
We ground our ethical assessment on four principles: 
(1)~\textit{Respect for Privacy}: we analyze purely public records without linking addresses to real identities, preventing user privacy violations; 
(2)~\textit{Risk-Benefit Balance}: we expose unstudied MEV dynamics to enable protocol hardening against centralized value extraction; 
(3)~\textit{Fairness and Equity}: we sample unbiased cross-chain blocks without targeting specific users to democratize MEV transparency; and 
(4)~\textit{Compliance}: all offline experiments follow security research standards to facilitate open community oversight.

\stitle{Harm Mitigations:} 
We enforce three concrete safeguards against potential misuse: 
(1)~\textit{Mitigating Automated Exploitation}: we publish only high-level architectural designs and accompany every uncovered MEV variant with protocol-level defensive remediation; 
(2)~\textit{Constraining Cross-Chain Abuse}: we present cross-chain strategy portability strictly as a risk-forecasting methodology while withholding turnkey execution pipelines; and 
(3)~\textit{Preventing De-anonymization}: we report all empirical findings at aggregate mechanism granularity without highlighting or tracking individual wallet addresses.

\stitle{Publication Rationale \& Tiered Release:} 
Our decision to publish rests on a rigorous risk-benefit evaluation. Because MEV extraction is currently dominated by closed, proprietary syndicates, open academic analysis provides essential public benefits through transparent defensive benchmarks and community oversight. To balance scientific reproducibility with safety, we implement a tiered artifact release scheme: the core agent reasoning framework and defensive analysis modules are fully open-sourced, whereas execution components directly repurposable for live mempool exploitation remain strictly restricted.

\section*{Open Science}
\stitle{Codebase:} To foster open science and reproducibility, we make our codebase and experimental scripts available at an anonymized repository: \url{https://anonymous.4open.science/r/Blockchain-MEV-CF60}.

\eat{
\stitle{Data Availability \& Reproduction:} 
Due to storage constraints, full on-chain states are not bundled directly. Instead, our repository includes end-to-end extraction and replay scripts compatible with standard JSON-RPC endpoints. Full trace reproduction requires access to Ethereum, BSC, and Base archive nodes. All artifacts will be permanently archived and released upon paper acceptance.
}

\section{Detailed MEV Bot Workflow}
\label{Appendix:3-stage-separation}

\subsection{Stage I: On-Chain Data Collector}
The collector ingests and translates heterogeneous blockchain states into candidate observations through three sequential phases. First, it performs streaming block scans via \texttt{eth\_getLogs} indexed by Swap event topics to prune non-candidate transactions prior to invoking intensive state queries. Second, unknown event-emitting addresses are verified dynamically via batched \texttt{eth\_call} probes (confirming valid factory pointers, distinct token pairs, and positive fee tiers) without static registries, while caching unsupported contracts to eliminate redundant probing. Third, for authenticated candidates, it decodes receipts and calldata to extract user-enforced execution bounds (e.g., slippage tolerances, input caps, deadlines, and recipients), including unwrapping nested swap paths within multi-hop aggregators (e.g., KyberSwap).

\subsection{Stage II: MEV Planner}
The planner compiles observation candidates into a deterministic intermediate representation (IR) that encodes execution paths, parameters, and ordering constraints without asserting runtime profits. It executes specialized convex optimization algorithms (detailed in Appx.~\ref{Appendix:CPMM_Strategy_Engineering_Details}) to compute optimal trade volumes and multi-pool routes. For ecosystems with private order flows (e.g., Ethereum), transactions are organized into a \texttt{bundle\_atomic} structure, relying on block builders to guarantee atomic ordering between searcher and victim transactions. Furthermore, the plan sequences participant roles (front-run, victim, and back-run), decomposing attacker transactions into modular action primitives (e.g., Swap, Add/Remove Liquidity across CPMM, CLMM, and Balancer~V2) with runtime-resolved parameters, while preserving raw payloads for victim replay.

\subsection{Stage III: Validator Runtime Simulation and Profit Calculation}
The validator serves as the authoritative engine for economic outcome determination, evaluating synthesized plans within high-fidelity sandboxes. Each candidate executes on an isolated node forked at the pre-target block under on-demand mining to eliminate state contamination. Runtime adapters compile semantic primitives into ABI-encoded calldata to expose parameter flaws, executing victim transactions via \texttt{eth\_sendRawTransaction} and falling back to sender impersonation upon nonce conflicts. Post-execution, the engine measures net balance deltas, crediting native assets directly, converting non-native tokens via CoinGecko pricing at the fork timestamp, and performing mark-to-pool offline liquidation for unclosed LP positions. To eliminate false positives, any pricing, decimal, or liquidation failure marks profit as null ($\text{profit}=\mathtt{null}$), with all validated cases archived into immutable \texttt{all-candidates}, \texttt{runtime-success}, and \texttt{profit-found} artifacts.
\section{Evaluated MEV Strategies Workflow}
\label{Appendix:CPMM_Strategy_Engineering_Details}
\label{Appendix:Baselines}
To perform a consistent framework benchmarking, we adopt the following idealization principles: The attacker may acquire any ERC-20 or native asset required by a strategy and may establish the necessary token approvals before execution. The gas costs of any approvals are also considered in the revenue calculation. Consequently, the attack is not constrained by the attacker's historical inventory, financing channels, or capital capacity. The attacker is also assumed to observe the target transactions before inclusion and to control the relative placement of its own transactions within the same block. We exclude searcher competition, builder payments, and network delays. These assumptions do not make every candidate profitable; rather, they allow us to evaluate all valid historical opportunities under a consistent setting. The strategies we adopt for evaluation are HFT, SBA, JIT, BUR, LR, MLL and their combined versions with HFT, which are as follows:
\subsection{HFT: High-Frequency Trading}

HFT reproduces a classic single-victim sandwich attack from~\cite{zhou2021high}. Leveraging the victim's swap direction and slippage limit, the attacker front-runs the victim swap to inflate the pool price, lets the victim execute at degraded slippage, and immediately back-runs the victim with a reverse swap to close the arbitrage.

The primary engineering challenge in HFT is optimizing the front-run input size. We resolve this via a three-step procedure: First, we employ a binary search to find the maximum feasible input volume bounded by the victim's slippage tolerance. Second, we verify net profitability at this maximum volume and terminate early if unprofitable, as no smaller input can yield positive returns under monotonic pool dynamics. Third, we perform a ternary search over the feasible input range to pinpoint the exact profit-maximizing trade size while accounting for EVM integer truncation artifacts.



\subsection{SBA: Swap Backrun Arbitrage}
SBA captures MEV via a single-pool backrun strategy that profits by trading in reverse against the price displacement induced by a victim swap. Leveraging the temporary reserve shift caused by the victim's exact-input trade, the searcher computes the post-execution pool state and immediately executes an opposite-direction swap to revert the marginal price toward its pre-transaction level, capturing the resulting execution spread. To optimize planning latency, the module adopts a strictly localized same-pool markout model, foregoing cross-pool path search and multi-hop routing during the planning phase.

To compute arbitrage profit, asset values must be normalized into the network's native currency (e.g., ETH on Ethereum). The SBA planner handles four asset categories accordingly: (1)~WETH, mapped directly at 1:1; (2)~WETH-paired assets, priced via on-chain reserves after the swap; (3)~Non-WETH-paired assets, converted using block-synchronized CoinGecko prices; and (4)~Unpriced assets, classified as \textit{pricing uncomputable} and dropped to filter out non-settleable opportunities.

\subsection{JIT: Just-in-Time Liquidity}
JIT captures MEV by sandwiching a large victim swap with transient liquidity to extract high trading fees without long-term holding risk. Detecting an impending large trade, the searcher front-runs it by injecting concentrated liquidity into the target pool, collects the dominant share of swap fees generated during execution, and immediately back-runs it to withdraw the position, profiting whenever the captured fee exceeds the inventory loss.

The planning module optimizes the liquidity size by balancing fee returns against inventory risk: under-sizing yields negligible fees, while over-sizing amplifies inventory displacement beyond fee gains. For each candidate size, the planner prices deposited and withdrawn assets in WETH using synchronized external feeds, computing net profit as the exit value minus the entry value to select the size with the highest positive increment.

To isolate the extraction potential of JIT and maintain a tractable simulation baseline, the evaluation operates under four simplifying assumptions: (1)~atomic transaction adjacency for liquidity mint and burn around the victim swap; (2)~unconstrained capital access with zero opportunity cost; (3)~no concurrent liquidity events or intervening swaps within the target block; and (4)~isolated execution without searcher competition.

\subsection{BUR: Burger Sandwich Attack}
BUR extends the single-victim sandwich attack in HFT to a sequence of same-direction swaps within the same pool and block. Leveraging the clustered victims' cumulative price impact, the attacker front-runs the batch with a single swap, allows the victims to execute sequentially at degraded slippage, and back-runs immediately after the final victim to close the arbitrage.

The primary engineering challenge is preserving the joint executability of the full victim sequence. An excessive front-run induces severe slippage displacement that risks triggering the victims' slippage protection and reverting the transactions. To address this, the planner searches over candidate attack sizes defined by a basis-point (bps) ladder of total victim volume, replaying the sequence against simulated post-front-run reserves to select the first size that keeps all trades executable while yielding positive net profit.

\subsection{LR: Liquidity Refined}
LR augments the sandwich attack by embedding a liquidity-provision step within a two-stage back-run sequence. Leveraging the victim's price displacement and transient fee yield, the searcher front-runs the swap, executes a partial back-run, injects liquidity to mint LP tokens, and submits a final back-run to capture the remaining spread while retaining the LP position.

The core challenge in LR is balancing capital allocation between directional arbitrage and liquidity minting, since misallocating funds to either leg reduces overall capital efficiency and terminal profit. The planner resolves this trade-off by testing discrete liquidity-addition ratios across the five-step pipeline, marking remaining LP tokens to WETH via external price feeds to select the configuration that maximizes aggregate terminal gains.

\subsection{MLL: Multi-Victim Liquidity Layering}
MLL extends JIT liquidity provision to batch-capture trading fees across a consecutive sequence of same-pool, same-direction swaps. Leveraging the aggregate volume across multiple trades, the searcher front-runs the batch with a temporary liquidity injection, collects proportional fee shares across all successive swaps, and back-runs the final trade to withdraw the position.

The fundamental challenge in MLL is sizing the liquidity injection against compounding inventory risk, as larger positions capture greater fee shares but amplify capital lockup and cumulative inventory loss across the batch. To optimize this trade-off, the planner simulates candidate liquidity ratios across the add-swap-remove lifecycle, using synchronized external prices to select the scale that maximizes net returns over inventory displacement.

\subsection{SBA-HFT: Swap Backrun Arbitrage with HFT Optimization}
SBA-HFT improves standard SBA by running a fine-grained numeric search for the optimal back-run trade size. When a victim swap shifts the target pool's price, the searcher back-runs the trade with an opposing swap to capture the price gap, without front-running or cross-pool routing. By expanding the search beyond plain heuristics, SBA-HFT captures extra profit while avoiding oversized trades whose self-induced price impact would erase the arbitrage. Like standard SBA, it uses external price benchmarks and assumes idealized capital access.

We outline the execution workflow in Algorithm~\ref{alg:sba_hft} using four building blocks. A victim transaction $T_v$ is a pending mempool trade that displaces pool prices, and the focal swap $S_{focal}$ is the exact internal hop that touches target pool $P$. An atomic bundle packages the attacker's back-run directly behind $T_v$ via block builders to prevent third-party interference. Finally, \textsc{ValidateOnFork} simulates the bundle on a local state fork to verify net profit after gas and prevent on-chain reverts.

The algorithm encapsulates trade-size optimization inside \textsc{SizeOptimization}. To navigate the non-linear profit curve efficiently, it applies a three-step search: coarse probing across available capital bounds, ternary search to isolate the unimodal peak, and granular local scanning to pinpoint the optimal input size $s^*$.

\subsection{JIT-HFT: HFT-Optimized Just-in-Time Liquidity}
JIT-HFT inherits standard JIT's profit model by front-running an impending victim swap with transient liquidity and immediately back-running it to harvest fees and inventory revaluations. Instead of relying on a single heuristic ratio, JIT-HFT searches across orders of magnitude to pinpoint the optimal liquidity scale.

The central challenge lies in navigating the non-linear trade-off between fee capture and inventory displacement: undersized positions capture negligible fees, while oversized positions amplify inventory loss beyond fee gains. To resolve this, JIT-HFT evaluates a ratio ladder spanning three orders of magnitude (2.5\%--200\% of victim volume). For each candidate, the planner simulates the add-swap-remove lifecycle, prices net asset differentials in ETH via external feeds, and selects the scale maximizing the net increment.

JIT-HFT retains the idealized assumption of unconstrained capital access at zero opportunity cost. This fine-grained search sharpens sizing resolution within single-victim, same-pool semantics, isolating mark-to-market returns without modeling capital costs, searcher competition, or multi-block retention.

\subsection{BUR-HFT: HFT-Optimized Burger Sandwich}

BUR-HFT preserves BUR's profit mechanism: several aligned victim swaps cumulatively move the pool price, and the attacker monetizes that movement with one position surrounding the group. Its only substantive extension is the method used to select the attack size.

The main engineering concern in BUR-HFT is the joint feasibility of the search. Because each victim changes the reserves observed by subsequent victims, every candidate input must keep the entire ordered victim group within its slippage constraint. We first perform range expansion to locate an upper bound at which the group breaks, then binary search that bound to find the largest feasible frontrun. We check profitability at this maximum input; if it is not profitable, no smaller input can be profitable and we stop. The protocol's integer division may let smaller inputs generate equal or higher profit, so we run ternary search over the feasible range and scan nearby integer points to locate the discrete optimum. BUR-HFT inherits BUR's assumptions of complete advance knowledge and unrestricted attack capital; HFT-style optimization improves sizing resolution without changing the victim-group definition.

\subsection{LR-HFT: HFT-Optimized Liquidity Refined}

LR-HFT preserves LR's combination of partial backrun, liquidity provision, and final backrun. Its profit comes from the same two sources as LR but it replaces fixed attack sizing with a joint search over the directional position and the liquidity configuration.

The main engineering concern in LR-HFT is that the two choices interact. The initial position determines both the price movement available for extraction and the attacker's post-victim inventory, while the liquidity configuration determines the LP position, the increase in pool depth, and the slippage improvement on the remaining backrun. We therefore evaluate them jointly rather than optimizing either in isolation. For each candidate frontrun amount we scan a coarse ladder of liquidity-add ratios, locate the strongest broad region, and then refine around it on a finer grid. We compare each complete LR portfolio with a sandwich baseline alternative that adds no liquidity and only count in the cases whose refinement creates incremental market value.

\subsection{MLL-HFT: HFT-Optimized Multi-Victim Liquidity Layering}

MLL-HFT preserves MLL's use of one temporary LP position across a same-block group of aligned victim swaps. Its profit comes from the same fee capture and LP revaluation as MLL, while it uses HFT-style search to optimize the liquidity scale.

The main engineering concern in MLL-HFT is the scale trade-off. A larger position receives a greater share of the victims' fees and reduces their price impact, but requires more two-sided capital and creates greater inventory-rebalancing exposure. We explore a broad range of liquidity scales, refine the most profitable region, and select the position whose withdrawn assets show the largest positive value increment over the initial deposit, valued under historical external prices.

\begin{algorithm}[htbp]
\caption{SBA-HFT: Single-Pool Swap-Backrun Strategy}
\label{alg:sba_hft}
\begin{algorithmic}[1]
\Require Victim transaction $T_v$, Target pool $P$, Max capital $C_{max}$
\Ensure Executed atomic bundle or $\emptyset$

\vspace{1.5mm}
\Statex \textbf{\textit{Phase 1: Simulation \& Setup}}
\State $S_{focal} \gets \text{ExtractTargetSwap}(T_v, P)$ 
\State $State_{post} \gets \text{SimulateImpact}(S_{focal}, P)$ 
\State $S_{rev} \gets \text{ConstructReverseSwap}(S_{focal})$ 

\vspace{1.5mm}
\Statex \textbf{\textit{Phase 2: Three-Stage SizeOptimization}}
\State $s_{max} \gets \min(C_{max}, \text{VictimVolume} \times 2)$

\Statex {\textbf{Stage 1: Coarse Probe}}
\State Probe sizes exponentially up to $s_{max}$ to verify a valid profit range.

\Statex {\textbf{Stage 2: Ternary Search}}
\State $W \gets [1, s_{max}]$ \Comment{Search window}
\While{$|W| > 64$}
    \State $p_{left} \gets \text{SimulateProfit}(\text{1/3 point of } W)$
    \State $p_{right} \gets \text{SimulateProfit}(\text{2/3 point of } W)$
    \If{$p_{left} < p_{right}$}
        \State $W \gets \text{Right 2/3 of } W$
    \Else
        \State $W \gets \text{Left 2/3 of } W$
    \EndIf
\EndWhile

\Statex {\textbf{Stage 3: Exhaustive Local Scan}}
\State $s^*, p^* \gets \text{Scan } W \text{ to find exact optimum }s^*$
\If{$p^* \leq 0$} \Return $\emptyset$ \EndIf

\vspace{1.5mm}
\Statex \textbf{\textit{Phase 3: Atomic Execution}}
\State $Bundle \gets \text{AtomicSequence}[T_v \rightarrow S_{rev}(s^*)]$ 
\State $NetProfit \gets \text{ValidateOnFork}(Bundle) - \text{GasCost}$
\If{$NetProfit > 0$}
    \State \Return $\text{Execute}(Bundle)$
\EndIf
\State \Return $\emptyset$

\end{algorithmic}
\end{algorithm}
\newtheorem{theorem}{Theorem}
\newtheorem{proposition}{Proposition}
\section{Theoretical Analysis of SBA-HFT}
\label{app:SBA-HFT_theoretical_proofs}

\subsection{Problem Setup and Profit Model}

Imagine the liquidity pool right after a target user (the victim) has traded. We call this post-trade state $S$. If our bot injects an input amount $a$, it extracts $F_S(a)$ in output tokens. Given the external market price $p$ and the network gas cost $G(a)$, our net WETH profit $\Pi_S(a)$ is straightforward: Revenue minus Token Cost minus Gas.
\begin{equation}
  \Pi_S(a) = F_S(a) - p\,a - G(a).
  \label{eq:runtime-objective}
\end{equation}

Both SBA and SBA-HFT strategies share this exact profit formula. The only difference is \textit{where} they look for the best input $a$:
\begin{itemize}
    \item \textbf{SBA search set ($\mathcal{B}$):} Uses heuristic shortcuts. It only evaluates a handful of sizes, calculated as explicit percentages of the victim's output $v_{\mathrm{out}}$.
    \item \textbf{SBA-HFT search set ($\mathcal{H}$):} Scans comprehensively. It evaluates every possible integer amount up to a safe mathematical upper limit.
\end{itemize}
Since SBA‑HFT explores all possibilities, its search space $\mathcal{H}$ completely absorbs SBA's smaller search space $\mathcal{B}$ (meaning $\mathcal{B} \subseteq \mathcal{H}$ is always true).

\subsection{Search-Set Dominance}

\begin{theorem}[Search-set dominance]
\label{thm:search-set-dominance}
Let $a_H$ be the best size found by SBA-HFT in $\mathcal{H}$, and $a_B$ be the best size found by SBA in $\mathcal{B}$. SBA-HFT will always achieve at least the same WETH profit as SBA: $\Pi_S(a_H) \geq \Pi_S(a_B)$.
\end{theorem}
\textit{Proof.} This follows trivially from the subset condition $\mathcal{B} \subseteq \mathcal{H}$, which guarantees that the maximum of $\Pi_S(a)$ over the superset $\mathcal{H}$ is greater than or equal to its maximum over the subset $\mathcal{B}$.

\subsection{Strict Profit Improvement}

\begin{proposition}[Strict dominance in real-world markets]
\label{prop:strict-dominance}
In live trading environments, SBA-HFT will almost always strictly outperform SBA: $\Pi_S(a_H) > \Pi_S(a_B)$.
\end{proposition}
\noindent\textit{Proof.} Let $a^* = \arg\max_{a} \Pi_S(a)$ denote the theoretical optimal input size. Its value depends strictly on the post-trade liquidity state $S$ and external market price $p$, remaining mathematically decoupled from the victim's trade volume $v_{\mathrm{out}}$. 


Conversely, the SBA strategy restricts its candidate search space to heuristic multiples of the victim's transaction, formalized as $\mathcal{B} = \{ c \cdot v_{\mathrm{out}} \mid c \in C \}$.  For SBA to blindly hit the exact optimal size, a random percentage of a random user's trade would have to perfectly collide with the market-driven mathematical peak $a^*$. Because these variables are driven by totally independent market forces, the probability of this exact collision is effectively zero.

Consequently, the true peak $a^*$ is almost never in SBA's limited playbook ($a^* \notin \mathcal{B}$). SBA-HFT's exhaustive search will therefore step beyond SBA's blind spots and reliably capture a strictly higher WETH profit. (This same logic applies directly to CLMMs like Uniswap V3, where optimal sizing depends on tick architecture, not the victim's volume.)
\section{Tick-Level State Transitions in SBA-HFT}
The victim moves upward through the 11 initialized ticks 197710, 197720,
\ldots, 197810.  Both backruns start at tick 197817 and first cross the same
three ticks 197810, 197800, 197790.  SBA stops at tick 197787.  SBA-HFT
continues down to tick 197600, traversing eighteen additional initialized
ticks (21 in total).  Table~\ref{tab:all-crossings} lists every initialized
tick crossed by each swap together with its liquidity net; these are the
values required to replay the state transitions independently.

\begin{table}[htbp]
\centering
\caption{All initialized ticks crossed, with liquidity net. The shared
backrun prefix is 197810, 197800, 197790; the SBA-HFT only complementary ticks begins at
197780.}
\label{tab:all-crossings}
\sisetup{
  scientific-notation = true,
  round-mode = places,
  round-precision = 4,       
  table-format = -1.4e2       
}
\resizebox{\columnwidth}{!}{%
\begin{tabular}{r S @{\quad} r S @{\quad} r S}
\toprule
\multicolumn{2}{c}{\textbf{Victim (11)}} & \multicolumn{2}{c}{\textbf{SBA (3)}} & \multicolumn{2}{c}{\textbf{SBA-HFT extra (18)}} \\
{\textbf{Tick}} & {\textbf{Liquidity net}} & {\textbf{Tick}} & {\textbf{Liquidity net}} & {\textbf{Tick}} & {\textbf{Liquidity net}} \\
\midrule
197710 & -119461249973742398  & 197810 & -31362008543507   & 197780 & -383366919250281 \\
197720 & -1081438842953574    & 197800 & 1286992685625811  & 197770 & -203362073447790 \\
197730 & 2229836423358469383  & 197790 & -20297835264190   & 197760 & -1115322058013891 \\
197740 & -2227345450299561462 &        &                   & 197750 & -2780159707162433 \\
197750 & -2780159707162433    &        &                   & 197740 & -2227345450299561462 \\
197760 & -1115322058013891    &        &                   & 197730 & 2229836423358469383 \\
197770 & -203362073447790     &        &                   & 197720 & -1081438842953574 \\
197780 & -383366919250281     &        &                   & 197710 & -119461249973742398 \\
197790 & -20297835264190      &        &                   & 197700 & -422782575204817 \\
197800 & 1286992685625811     &        &                   & 197690 & -449999530677936423 \\
197810 & -31362008543507      &        &                   & 197680 & 683986602455997 \\
       &                      &        &                   & 197670 & -4376644206972401 \\
       &                      &        &                   & 197660 & 270485275060007 \\
       &                      &        &                   & 197650 & 4126567458044491 \\
       &                      &        &                   & 197640 & 197309977538567 \\
       &                      &        &                   & 197630 & 162995672188181 \\
       &                      &        &                   & 197620 & -1708193886200303 \\
       &                      &        &                   & 197610 & 167583061487984 \\
\bottomrule
\end{tabular}%
}
\end{table}

\subsection{Replay of the SBA-HFT Backrun}
\label{sec:worked-replay}

This subsection recomputes the SBA-HFT backrun from
Table~\ref{tab:all-crossings} alone, using only the within-tick invariant of
Eq.~\ref{eq:clmm-output} and the tick-crossing liquidity update.  All
quantities are raw on-chain integers; $\sqrt{P}$ denotes \texttt{sqrtPriceX96}
and $2^{96}$ is the fixed-point scale.  Because USDC has 6 decimals
and WETH has 18, the on-chain $\sqrt{P}$ already embeds the raw-unit
price ratio, so the tick math below uses $L$ directly with no additional
decimal scaling.

\paragraph{Per-interval step.}
For a (USDC in, WETH out) exact-input swap moving the price
downward inside one tick interval with active liquidity $L$ and current
square-root price $s=\sqrt{P}$, the next initialized tick boundary is
$s_{\mathrm{next}}=\sqrt{1.0001^{\,t_{\mathrm{next}}}}\cdot 2^{96}$.  The net
input required to reach that boundary and the WETH output of traversing the
whole interval are
\begin{equation}
\Delta a_{\mathrm{net}}
  = \frac{L\,2^{96}\,(s-s_{\mathrm{next}})}{s\,s_{\mathrm{next}}},
\qquad
\Delta y
  = \frac{L\,(s-s_{\mathrm{next}})}{2^{96}}.
\label{eq:interval-step}
\end{equation}
Charging the $f=0.05\%$ fee on input, the gross input for the full interval is
$\Delta a_{\mathrm{gross}}=\lceil \Delta a_{\mathrm{net}}/(1-f)\rceil$.  If the
remaining input is smaller, the swap ends inside the interval at
\begin{equation}
s' = \frac{L\,2^{96}\,s}{L\,2^{96}+(1-f)\,\Delta a_{\mathrm{in}}\,s},
\qquad
\Delta y = \frac{L\,(s-s')}{2^{96}}.
\label{eq:partial-step}
\end{equation}
On crossing tick $t$ while moving downward, liquidity updates as
$L \leftarrow L - \mathrm{liquidityNet}(t)$.

\paragraph{Substitution check.}
Start from the post-victim state $\sqrt{P}$ and $L$, with input $1{,}904{,}210$ USDC. Applying Eq.~\ref{eq:interval-step} interval by interval:

\begin{table}[htbp]
\centering
\caption{Replay of the SBA-HFT backrun. Initial post-victim state: $\sqrt{P}=1563942608515956339823549149976335$, $L=3097836570594546384$, tick $=197817$. Step~0 is the starting state; steps~1--21 cross the listed initialized ticks; step~22 is the final partial interval ending at tick~197600.}
\label{tab:hft-replay-steps}
\resizebox{\columnwidth}{!}{%
\begin{tabular}{
  r 
  r 
  S[input-ignore={,}, scientific-notation=true,  round-mode=places, round-precision=4, table-format=1.2e1] 
  S[input-ignore={,}, scientific-notation=false, round-mode=places, round-precision=4, table-format=2.3] 
  S[input-ignore={,}, scientific-notation=true,  round-mode=places, round-precision=4, table-format=1.4e2]
}
\toprule
\textbf{Step} & \textbf{Cross tick} & {\textbf{Interval input (USDC)}} & {\textbf{Interval output (WETH)}} & {\textbf{Liquidity after}} \\
\midrule
1  & 197810 & 60,174.545853  & 23.426730 & 3097867932603089891 \\
2  & 197800 & 78,552.886057  & 30.554659 & 3096580939917464080 \\
3  & 197790 & 78,559.519659  & 30.526699 & 3096601237752728270 \\
4  & 197780 & 78,599.322484  & 30.511640 & 3096984604671978551 \\
5  & 197770 & 78,648.365664  & 30.500165 & 3097187966745426341 \\
6  & 197760 & 78,692.864704  & 30.486921 & 3098303288803440232 \\
7  & 197750 & 78,760.571105  & 30.482655 & 3101083448510602665 \\
8  & 197740 & 78,870.667798  & 30.494757 & 5328428898810164127 \\
9  & 197730 & 135,587.101394 & 52.371350 & 3098592475451694744 \\
10 & 197720 & 78,886.157015  & 30.439808 & 3099673914294648318 \\
11 & 197710 & 78,953.153786  & 30.435212 & 3219135164268390716 \\
12 & 197700 & 82,037.009881  & 31.592384 & 3219557946843595533 \\
13 & 197690 & 82,088.816245  & 31.580739 & 3669557477521531956 \\
14 & 197680 & 93,609.209522  & 35.976806 & 3668873490919075959 \\
15 & 197670 & 93,638.566488  & 35.952120 & 3673250135126048360 \\
16 & 197660 & 93,797.153594  & 35.977016 & 3672979649850988353 \\
17 & 197650 & 93,837.151203  & 35.956385 & 3668853082392943862 \\
18 & 197640 & 93,778.601044  & 35.898036 & 3668655772415405295 \\
19 & 197630 & 93,820.453812  & 35.878162 & 3668492776743217114 \\
20 & 197620 & 93,863.202964  & 35.858635 & 3670200970629417417 \\
21 & 197610 & 93,953.872193  & 35.857400 & 3670033387567929433 \\
22 & (partial, end tick 197600) & 85,501.648427 & 32.600477 & 3670033387567929433 \\
\midrule
\multicolumn{2}{r}{\textbf{Total}} & \textbf{1,904,210.840892} & \textbf{733.358755} & \multicolumn{1}{c}{} \\
\bottomrule
\end{tabular}%
}
\end{table}

The first 21 rows exhaust the listed crossings; the remaining $85{,}501$ USDC is consumed by Eq.~\ref{eq:partial-step} in the final interval and ending at tick $197600$.  Summing the interval outputs gives $733.36$WETH.

\begin{table}[hbtp]
\centering
\caption{Profit distribution of JIT-HFT on CLMM. \texttt{P25--P75} denotes the interquartile range of per-case capital cost.}
\label{tab:jit-hft-v3-profit-dist}
\scriptsize
\setlength{\tabcolsep}{2.5pt}
\begin{tabularx}{\columnwidth}{@{} l r r r r r @{}}
\toprule
\textbf{Profit bucket} & \textbf{Cases} & \textbf{Share} & \textbf{Avg. capital} & \textbf{Capital P25--P75} & \textbf{Profit share} \\
\textbf{(WETH)} & & \textbf{(\%)} & \textbf{(WETH)} & \textbf{(WETH)} & \textbf{(\%)} \\
\midrule
\midrule
$< 0.0001$     & 817{,}868 & 43.63 &    1.25 & $0.006$--$0.223$      &  0.04 \\
$0.0001$--$0.001$ & 358{,}863 & 19.14 &    5.41 & $0.202$--$1.891$      &  0.46 \\
$0.001$--$0.01$   & 404{,}901 & 21.60 &   12.20 & $0.812$--$7.458$      &  4.85 \\
$0.01$--$0.1$     & 239{,}251 & 12.77 &   34.96 & $3.967$--$32.133$     & 24.18 \\
$0.1$--$1$        &  49{,}882 &  2.66 &  114.55 & $21.332$--$126.026$   & 39.54 \\
$1$--$10$         &   3{,}526 &  0.19 &  720.56 & $122.511$--$548.633$  & 24.82 \\
$\geq 10$         &       111 &  0.01 & 4{,}259.29 & $912.687$--$3{,}705.166$ &  6.10 \\
\bottomrule
\end{tabularx}
\end{table}

\begin{table}[H]
\centering
\small
\caption{Primary JPlag Matching Regions}
\label{tab:matching_regions}
\begin{tabularx}{\columnwidth}{l X}
\toprule
\textbf{Strategy} & \textbf{Corresponding Structure} \\
\midrule
HFT & Ternary search after the front-run upper bound, local scan, optimal profit selection, and search trace. \\
\addlinespace
JIT & JIT data collector, planner, and validator. Inspection of liquidity seed/valuation results, organization of the add liquidity and remove liquidity actions. \\
\addlinespace
SBA & SBA planner. Single-victim opportunity normalization, same-pool back-run sizing, execution-ready/rejection states. \\
\addlinespace
JIT-HFT & Ternary search within a restricted interval, simulation success gate, sizing fields, and plan workflow. Protocol-specific parts (V3 tick/range, Balancer liquidity ratio) are still implemented within each adaptation. \\
\addlinespace
SBA-HFT & SBA-HFT planner containing probe upper bound, ternary search, local scan, optimal back-run selection, and workflow summary. This family has high JPlag but zero JSCPD, which perfectly exemplifies that the internal logic are preserved, while protocol implementation rewritten. \\
\bottomrule
\end{tabularx}
\end{table}

\section{Case Category Definitions}
\label{four_case_categories}
We distinguish four outcome categories along the three-stage MEV pipeline. 
A candidate is a structured record of one or more historical transactions and the corresponding block-level state that matches a predefined MEV strategy pattern and is admitted by the data collector (Tab.~\ref{fig:agent_loop_full}) for downstream processing. A candidate becomes planned only when the planner constructs a complete executable plan containing the required transactions, routes, parameters, and execution constraints. Runtime success is recorded only when the Validator successfully simulates that plan at the historical fork point. Finally, profit found is a strictly stronger condition: the validator must deterministically normalize the realized revenue, subtract all gas costs, and confirm a realized net profit greater than zero.

Each planned candidate is executed on an isolated anvil fork at the corresponding historical block. Profitability is computed from attacker asset-balance deltas reducing gas costs, and finally normalized to the respective native token of each chain (i.e., ETH for Ethereum and Base, BNB for BSC). The historical daily token prices are derived from the CoinGecko API and locally saved as a database. Missing historical prices are marked as uncomputable rather than being treated as zero or negative profit.

\end{document}